\documentclass{article}

\usepackage[T1]{fontenc}
\usepackage[utf8]{inputenc}

\usepackage{amsmath,amssymb}
\usepackage{geometry}
\usepackage{algorithm,algpseudocode}

\usepackage{xcolor}
\usepackage{subcaption}
\usepackage{hyperref}
\usepackage{amsthm}

\newcommand{\soft}[1]{\texttt{#1}}
\newcommand{\ft}{\soft{feyntrop} }
\newcommand{\eq}[1]{\begin{align} #1 \end{align}}

\definecolor{col1}{RGB}{254,97,0}
\definecolor{col2}{RGB}{100,143,255}
\definecolor{col3}{RGB}{120, 94, 240}
\definecolor{col4}{RGB}{220, 38, 127}
\definecolor{col5}{RGB}{255, 176, 0}

\newtheorem{theorem}{Theorem}
\numberwithin{theorem}{section}

\theoremstyle{definition}

\theoremstyle{remark}

\newtheorem{assumption}[theorem]{Assumption}
\newtheorem{observation}[theorem]{Observation}
\newtheorem{conjecture}[theorem]{Conjecture}

\newcommand{\RP}{\mathbb{RP}}
\newcommand{\RR}{\mathbb{R}}
\newcommand{\PP}{\mathbb{P}}
\newcommand{\CC}{\mathbb{C}}
\newcommand{\CP}{\mathbb{CP}}
\newcommand{\ZZ}{\mathbb{Z}}

\newcommand{\cU}{\mathcal{U}}
\newcommand{\cF}{\mathcal{F}}

\newcommand{\cV}{\mathcal{V}}

\newcommand{\cI}{\mathcal{I}}
\newcommand{\tr}{\mathrm{tr}}

\newcommand{\dd}{\mathrm{d}}
\newcommand{\xx}{\boldsymbol{x}}
\newcommand{\XX}{\boldsymbol{X}}

\newcommand{\supp}{\operatorname{supp}}
\newcommand{\newt}{\mathbf{N}}

\usepackage{tikz}    
\usetikzlibrary{positioning}

\usepackage{tikz-feynman}
\tikzfeynmanset{compat=1.1.0}

\hypersetup{
  pdfauthor={Michael Borinsky, Henrik J. Munch, Felix Tellander},
  pdftitle={Tropical Feynman integration in the Minkowski regime},
  colorlinks=true,
  linkcolor  = col3,
  citecolor  = col4,
  urlcolor   = col3,
}

\title{Tropical Feynman integration in the Minkowski regime}%

\author{Michael Borinsky%
\\ 
\parbox{.45\linewidth}{
\center \small \textsc
Institute for Theoretical Studies\\
ETH Zürich\\
8092 Zürich, Switzerland
}
\and
Henrik J. Munch%
\\
\parbox{.45\linewidth}{
\center \small \textsc
Dipartimento di Fisica e Astronomia\\
Università degli Studi di Padova\\
35131 Padova, Italy
}
\and
Felix Tellander%
\\
\parbox{.45\linewidth}{
\center \small \textsc
Deutsches Elektronen-Synchrotron DESY\\
Notkestr.~85\\
22607 Hamburg, Germany
}
}

\date{}%

\begin{document}

\vspace*{-2\baselineskip}%
\hspace*{\fill} \mbox{\footnotesize{\textsc{DESY-23-026}}}

{\let\newpage\relax\maketitle}

\begin{abstract}
We present a new computer program, \texttt{feyntrop}, which uses the tropical geometric approach to evaluate Feynman integrals numerically.
In order to apply this approach in the physical regime, we introduce a new parametric representation of Feynman integrals that implements the causal $i\varepsilon$ prescription concretely while retaining projective invariance. 
\texttt{feyntrop} can efficiently evaluate  dimensionally regulated, quasi-finite Feynman integrals, with not too exceptional kinematics in the physical regime, with a relatively large number of propagators and with arbitrarily many kinematic scales. We give a systematic classification of all relevant kinematic regimes, review the necessary mathematical details of the tropical Monte Carlo approach, give fast algorithms to evaluate (deformed) Feynman integrands, describe the usage of \texttt{feyntrop} and discuss many explicit examples of evaluated Feynman integrals.
\end{abstract}

\tableofcontents

\section{Introduction}
Feynman integrals are a key tool in quantum field theory. They are necessary to produce accurate predictions from given theoretical input such as a Lagrangian.  Applications are, for instance, the computations of virtual contributions to scattering cross-sections for particle physics phenomenology \cite{Heinrich:2020ybq}, corrections to the magnetic moment of the muon or the half-life of positronium \cite{Karshenboim:2005iy}, critical exponents in statistical field theory \cite{zinn2021quantum} and corrections to the Newton potential due to general relativity \cite{Donoghue:1994dn}.
An entirely mathematical application of Feynman integrals is the certification of cohomology classes in moduli spaces of curves or of graphs \cite{Brown:2021umn}.

In this paper, we introduce \texttt{feyntrop}\footnote{\ft can be downloaded from \url{https://github.com/michibo/feyntrop}.}, a new tool to evaluate Feynman integrals numerically.
In contrast to existing tools, \texttt{feyntrop} can efficiently evaluate Feynman integrals with a relatively large number of propagators and with an arbitrary number of scales. Moreover, \texttt{feyntrop} can deal with Feynman integrals in the \emph{physical} Minkowski regime and automatically takes care of the usually intricate contour deformation procedure. The spacetime dimension is completely arbitrary and integrals that are expanded in a dimensional regulator can be evaluated. The main restriction of \texttt{feyntrop} is that it cannot deal with Feynman integrals having subdivergences, that means the input Feynman integrals are required to be \emph{quasi-finite}. Moreover, \texttt{feyntrop} is not designed to integrate Feynman integrals at certain highly exceptional kinematic points. Outside the Euclidean regime, the external kinematics are required to be sufficiently \emph{generic}. 
It is worthwhile mentioning though that such highly exceptional kinematic points seem quite rare and \texttt{feyntrop} performs surprisingly well in these circumstances---in spite of the lack of mathematical guarantees for functioning. In fact, we were not able to find a quasi-finite integral with exceptional kinematics for which the integration with \texttt{feyntrop} fails. We only observed significantly decreased rates of convergence in such cases.

The mathematical theory of Feynman integrals has advanced rapidly in the last decades. 
Corner stone mathematical developments for Feynman integrals were, for instance, the systematic exploitation of their unitarity constraints (see, e.g., \cite{Bern:1994zx,Bern:2005cq}), the systematic solution of their integration-by-parts identities (see, e.g., \cite{Chetyrkin:1981qh,Laporta:2000dsw}), the application of modern algebraic geometric and number theoretic tools for the benefit of their evaluation (see, e.g., \cite{Bloch:2005bh,Brown:2008um,Panzer:2014caa}) and the systematic understanding of the differential equations which they fulfill (see, e.g., \cite{Remiddi:1997ny,Henn:2013pwa}).

Primarily, these theoretical developments were aimed at facilitating the \emph{analytic} evaluation of Feynman integrals. 
All known analytic evaluation methods are inherently limited to a specific class of sufficiently simple diagrams. Especially for high-accuracy collider physics phenomenology, such analytic methods are often not sufficient to satisfy the demand for Feynman integral computations at higher loop order,
which frequently involve complicated kinematics with many scales.
Even if an analytic expression for a given Feynman integral is available, it is usually a highly non-trivial task to perform 
the necessary analytic continuation into the physical kinematic regime.
On a different tack, computations of corrections to the Newton potential in the post-Newtonian expansion of general relativity \cite{Donoghue:1994dn}
require the evaluation of large amounts of Feynman diagrams in three dimensional Euclidean space. As analytic evaluation is often more difficult in odd-dimensional spacetime, tropical Feynman integration is a promising candidate to fulfill the high demand for large loop order Feynman integrals in this field.

For this reason, numerical methods for the evaluation of Feynman integrals seem unavoidable once a certain threshold in precision has to be overcome. 
In this paper, we will use \emph{tropical sampling} that was introduced in \cite{Borinsky:2020rqs} to evaluate Feynman integrals numerically.
This numerical integration technique is faster than traditional methods because the known (tropical) geometric structures of Feynman integrals are employed for the benefit of their numerical evaluation.
For instance, general Euclidean Feynman integrals with up to 17 loops and 34 propagators can be evaluated using basic hardware with the proof-of-concept implementation that was distributed by the first author with \cite{Borinsky:2020rqs}. The code of \ft is based on this implementation.
 The relevant mathematical structure is the \emph{tropical geometry} of Feynman integrals in the parametric representation \cite{Panzer:2019yxl,Borinsky:2020rqs}. This tropical geometry itself is a simplification of the intricate \emph{algebraic geometry} Feynman integrals display (see, e.g., \cite{Brown:2015fyf}).
Tropical Feynman integration was already used, for instance, in \cite{Dunne:2021lie} to estimate the $\phi^4$ theory $\beta$ function up to loop order 11.
Some ideas from \cite{Borinsky:2020rqs} were already implemented in the \soft{FIESTA} package \cite{Smirnov:2021rhf}.
Tropical sampling was extended to toric varieties with applications to Bayesian statistics \cite{Borinsky:2022}. Moreover, the tropical approach was recently applied to study infrared divergences of Feynman integrals in the Minkowski regime \cite{Arkani-Hamed:2022cqe}.

The tropical approach to Feynman integrals falls in line with the increasing number of fruitful applications of tools from convex geometry in the context of quantum field theory. These include, for example, the discovery of polytopes in amplitudes (see, e.g. \cite{Arkani-Hamed:2013jha,Arkani-Hamed:2017tmz}). Further, Feynman integrals can be seen as generalized Mellin-transformations \cite{Nilsson:2013,Berkesch:2014,Schultka:2018nrs}. As such they are solutions to GKZ-type differential equation systems \cite{Gelfand1990}. Tropical and convex geometric tools are central to this analytic approach towards Feynman integrals (see, e.g., \cite{delaCruz:2019skx,Klausen:2019hrg,Klemm:2019dbm,Chestnov:2022alh,Tellander2022}). %

Tropical Feynman integration is closely related to the \emph{sector decomposition} approach \cite{Binoth:2000ps,Bogner:2007cr,Kaneko:2009qx}, which applies to completely general algebraic integrals. State of the art implementations of sector decompositions are, for instance, \soft{pySecDec} \cite{Borowka:2017idc} and \soft{FIESTA} \cite{Smirnov:2021rhf}. Other numerical methods that are tailored specifically to Feynman integrals are, for instance, difference equations \cite{Laporta:2000dsw}, unitarity methods \cite{Soper:1999xk}, the Mellin-Barnes representation \cite{Anastasiou:2005cb} and loop-tree duality \cite{Catani:2008xa,Capatti:2019ypt}. With respect to potential applications to collider phenomenology, the latter three have the advantage of being inherently adapted to Minkowski spacetime kinematics. A newer technique is the systematic semi-numerical evaluation of Feynman integrals using differential equations \cite{Liu:2017jxz,Mandal:2018cdj}, which is implemented, for instance, in \soft{AMFlow} \cite{Liu:2022chg}, \soft{DiffExp} \cite{Hidding:2020ytt} and \soft{SeaSyde} \cite{Armadillo:2022ugh}. A similar semi-numerical approach was put forward in \cite{Dubovyk:2022frj}. This technique can evaluate Feynman integrals quickly in the physical regime with high accuracy. 
A caveat is that it  relies on 
the algebraic solution of the usually intricate integration-by-parts system associated to the respective Feynman integral
and (usually) on
analytic boundary values for the differential equations 
(see \cite{Liu:2022chg,Liu:2022mfb} for an exception where the boundary values are computed exclusively from algebraic input). 
 We expect \soft{feyntrop}, which does not rely on any analytic or algebraic input, to be useful for computing boundary values as input for such methods.

\ft uses the \emph{parametric representation} of Feynman integrals for the numerical evaluation, which we briefly review in Section~\ref{sec:def}. This numerical evaluation has quite different characters in separate \emph{kinematic regimes}. We propose a new classification of such kinematic regimes in Section~\ref{sec:kinematic_regimes} which, in addition to the usual Euclidean and Minkowski regimes, includes the intermediate \emph{pseudo-Euclidean} regime. The original tropical Feynman integration implementation from \cite{Borinsky:2020rqs} was limited to the Euclidean regime.
Here, we achieve the extension of this approach to non-Euclidean regimes.

In the Minkowski regime, parametric Feynman integrands can have a complicated pole structure inside the integration domain. For the numerical integration an explicit \emph{deformation} of the integration contour, which respects the desired causality properties, is needed.
The use of explicit contour deformation prescriptions for numerics was 
pioneered in \cite{Soper:1999xk} and was later applied in the sector decomposition framework \cite{Binoth:2005ff}.
(Recently, a momentum space based approach for the solution of the deformation problem was put forward \cite{Pittau:2021jbs}.)
In Section~\ref{sec:contour}, we propose an explicit deformation prescription which, in its basic form, was employed in \cite{Mizera:2021icv} in the context of cohomological properties of Feynman integrals. This deformation prescription has the inherent advantage of retaining the projective symmetry of the parametric Feynman integrand. We provide explicit formulas for the Jacobian and thereby propose a new \emph{deformed parametric representation} of the Feynman integral. 

It is often desirable to evaluate a Feynman integral using dimensional regularization by adding a formal expansion parameter to the spacetime dimension, e.g.~$D=D_0-2\epsilon$, where $D_0$ is a fixed number and we wish to evaluate the Laurent or Taylor expansion of the integral in $\epsilon$. 
We will explain how \ft deals with such dimensionally regularized Feynman integrals in Section~\ref{sec:dimreg}. 
Moreover, we will discuss one of the major limitations of \ft in this section: 
In its present form \ft can only integrate Feynman integrals that are \emph{quasi-finite}. That means, input Feynman integrals are allowed to have an overall divergence, but no subdivergences. Further analytic continuation prescriptions (along the lines of \cite{Nilsson:2013,Berkesch:2014,vonManteuffel:2014qoa}) would be needed to deal with such subdivergences and we postpone the implementation of such prescriptions into \ft to a future publication. For now, the user of the program is responsible to render all input integrals quasi-finite; for instance by projecting them to a quasi-finite basis \cite{vonManteuffel:2014qoa}.
Note, however, that within our approach, the base dimension $D_0$ is completely arbitrary and can even be a non-integer value if desired.
The applicability in the case $D_0 =3$ makes \ft a promising tool for the computation of post-Newtonian corrections to the gravitational potential \cite{Kol:2007bc}.

In Sections~\ref{sec:tropicalapprox} and \ref{sec:tropsampling}, we will review the necessary ingredients for the tropical Monte Carlo approach from \cite{Borinsky:2020rqs}: The concepts of the \emph{tropical approximation} and \emph{tropical sampling}. In Section~\ref{sec:genperm}, we review the (tropical) geometry of parametric Feynman integrands and the particular shape that the Symanzik polynomials' Newton polytopes exhibit. We will put special focus on the \emph{generalized permutahedron} property of the second Symanzik $\cF$ polynomial. At particularly exceptional kinematic points, this property of the $\cF$ polynomial can be lost. In these cases the integration with \ft might fail. We discuss this limitation in detail in Section~\ref{sec:genperm}.
The overall tropical sampling algorithm is summarized in Section~\ref{sec:gentropsampling}.

In Section~\ref{sec:evaluation}, we summarize the necessary steps for the efficient evaluation of (deformed) parametric Feynman integrands.  The key step is to express the entire integrand in terms of explicit matrix expressions.  Our method is more efficient than the naive expansion of the Symanzik polynomials, as fast linear algebra routines can be used for the evaluation of such matrix expressions.

The structure, installation and usage of the program \ft is described in Section~\ref{sec:program}. To illustrate its capabilities we give multiple detailed examples of evaluated Feynman integrals in Section~\ref{sec:examples}.
In Section~\ref{sec:conclusions}, we conclude and give pointers for further developments of the general tropical Feynman integration method and the program \soft{feyntrop}.

\section{Feynman integrals}
\subsection{Momentum and parametric representations}
\label{sec:def}
Let $G$ be a one-particle irreducible Feynman graph with edge set $E$ and vertex set $V$.
Each edge $e \in E$ comes with a mass $m_e$ and an edge weight $\nu_e$. Each vertex $v \in V$ comes with 
an incoming spacetime momentum $p_v$. Vertices without incoming momentum, i.e.~where $p_v=0$, are internal.
Let $\mathcal E$ by the incidence matrix of $G$ 
which is formed by choosing an arbitrary orientation for the edges 
and setting $\mathcal E_{v,e} = \pm 1$ if $e$ points to/from $v$ and $\mathcal E_{v,e}= 0$ if 
$e$ is not incident to $v$.
The Feynman integral associated to $G$ reads%
\begin{equation}\label{eq:prop}
    \cI=
\int\prod_{e\in E}\frac{\dd^Dq_e}{i\pi^{D/2}}
\left(\frac{-1}{q_e^2-m_e^2+i\varepsilon}\right)^{\nu_e}
\prod_{v\in V\setminus \{v_0\}} i \pi^{D/2} \delta^{(D)}\left(p_v+\sum_{e\in E}\mathcal{E}_{v,e}q_e\right),
\end{equation}
where we integrate over all $D$-dimensional spacetime momenta $q_e$ and we extracted the $\delta$ function that accounts for overall momentum conservation by removing the vertex $v_0 \in V$.
We compute the squared length $q_e^2 = (q_e^0)^2- (q_e^1)^2 -(q_e^2)^2 -\ldots$ using the mostly-minus signature Minkowski metric.

To evaluate $\mathcal I$ numerically, we will use the equivalent parametric representation (see, e.g., \cite{Nakanishi:1971})
\begin{equation}\label{eq:par}
    \cI=\Gamma(\omega)\int_{\PP_+^{E}} \phi \quad
\text{ with }\quad
\phi = 
\left(\prod_{e\in E}\frac{x_e^{\nu_e}}{\Gamma(\nu_e)}\right)
\frac{1}{\cU(\xx)^{D/2}}\left(\frac{1}{\cV(\xx)-i\varepsilon \,\sum_{e\in E} x_e}\right)^{\omega}
\Omega\,.
\end{equation}
We integrate over the \emph{positive projective simplex}
$\PP_+^{E} = \{ \xx = [x_0, \ldots, x_{|E|-1}] \in \RP^{E-1} : x_e > 0 \}$
with respect to its canonical volume form
\begin{align} \Omega=\sum_{e=0}^{|E|-1}(-1)^{|E|-e-1}\frac{\dd x_0}{x_0}\wedge\cdots\wedge\widehat{\frac{\dd x_e}{x_e}}\wedge\cdots\wedge\frac{\dd x_{|E|-1}}{x_{|E|-1}} \, . \end{align}
Note that in the scope of this article we make the unusual choice to start the indexing with $0$ for the benefit of a seamless notational transition to our computer implementation.
So, the edge and vertex sets are always assumed to be given by $E = \{0,1,\ldots,|E|-1\}$ and $V = \{0,1,\ldots,|V|-1\}$.

The \emph{superficial degree of divergence} of the graph $G$ is 
given by $\omega = \sum_{e \in E} \nu_e - D L/2$, where $L = |E|-|V|+1$ is the number of loops of $G$.

We use $\cV(\xx) = \cF(\xx)/\cU(\xx)$ as a shorthand for the quotient of the 
two \emph{Symanzik polynomials} that can be defined using the \emph{reduced graph Laplacian} $\mathcal L(\xx)$,
a
$(|V|-1)\times(|V|-1)$ matrix given element-wise by 
    $\mathcal{L}(\xx)_{u,v}=\sum_{e\in E} \mathcal E_{u,e}\mathcal E_{v,e}/x_e $
for all $u,v \in V \setminus \{v_0\}.$
We have the identities
\begin{align} \label{eq:laplaceUF} \cU(\xx)&=\det \mathcal{L}(\xx) \,\left( \prod_{e\in E}x_e \right)\, ,& \cF(\xx)&=\cU(\xx)\,\left(-\sum_{u,v \in V\setminus\{v_0\}} \mathcal P^{u,v} ~ \mathcal{L}^{-1}(\xx)_{u,v}+\sum_{e\in E}m_e^2x_e\right), \end{align}
where $\mathcal P^{u,v} = p_u \cdot p_v$ with the scalar product being computed using the Minkowski metric.

\paragraph{Combinatorial Symanzik polynomials}
We also have the combinatorial formulas for $\cU$ and $\cF$
\begin{align} \label{eq:polyUF} \cU(\xx)&=\sum_{T} \prod_{e\notin T} x_e \, ,& \cF(\xx)&=-\sum_{F} p(F)^2 \prod_{e \notin F}x_e+\cU(\xx)\sum_{e\in E}m_e^2x_e \, , \end{align}
where we sum over all spanning trees $T$ and all spanning two-forests $F$ of $G$, and $p(F)^2$ is the Minkowski squared momentum running between the two-forest components. From this formulation it can be seen that $\cU$ and $\cF$ are homogeneous polynomials of degree $L$ and $L+1$ respectively. Hence, $\cV$ is a homogeneous rational function of degree $1$.

We will give fast algorithms to evaluate $\cU(\xx)$ and $\cF(\xx)$ in Section~\ref{sec:evaluation}.

\subsection{Kinematic regimes}
\label{sec:kinematic_regimes}

\begin{figure}

\centering

\newcommand{\xs}{1}

\begin{tikzpicture} \coordinate (v) at (0,0); \coordinate (w) at (0,\xs); \coordinate (u) at (0,{2*\xs}); \draw[fill=col4,opacity=.3] (v) ellipse ({2.3*2.5*\xs} and {2.5*\xs}); \draw[fill=col3,opacity=.3] (u) ellipse ({2.3*\xs} and {\xs}); \draw[fill=col2,opacity=.3] (w) ellipse ({2.3*2*\xs} and {2*\xs}); \draw[fill=col1,opacity=.3] (v) ellipse ({2.3*3*\xs} and {3*\xs}); \node (U) at (u) {Euclidean}; \node (W) at (v) {pseudo-Euclidean}; \node (G) at (0,-1.8*\xs) {generic}; \node (N) at (0,-2.75*\xs) {exceptional}; \end{tikzpicture}

\caption{Partition of kinematics into different regimes.}
\label{fig:kinematics}

\end{figure}
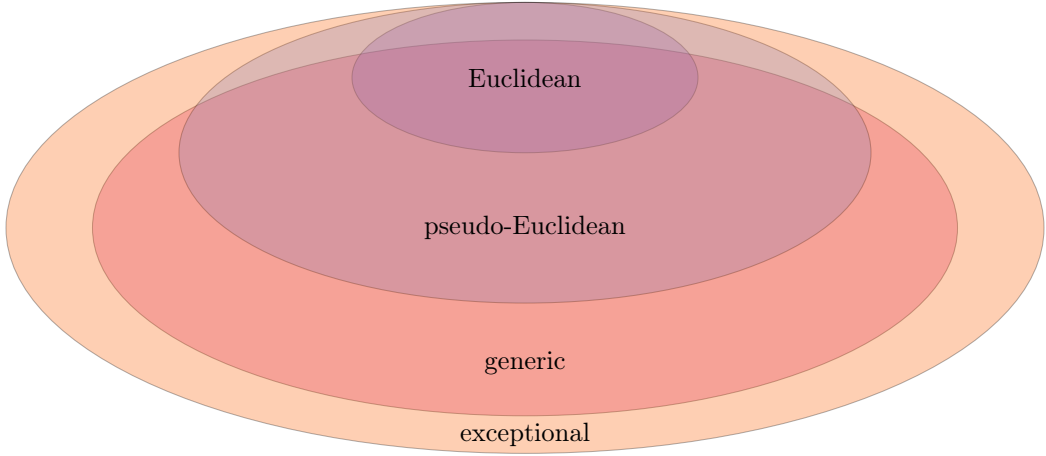

By Poincar\'e invariance,
the value of the Feynman integral \eqref{eq:prop}
only depends on the $|V|\times |V|$ Gram matrix $\mathcal P^{u,v} = p_u \cdot p_v$ and not on the explicit form of the vectors $p_v$. 
In fact, it is even irrelevant in which ambient dimension the vectors $p_v$ are defined.
The following characterisation of the different \emph{kinematic regimes} that we propose
will therefore only take the input of a symmetric $|V|\times |V|$ matrix $\mathcal P$
with vanishing row and column sums (i.e.~the momentum conservation conditions $\sum_{v \in V} p_u \cdot p_v = \sum_{v \in V} \mathcal P^{u,v} = 0$ for all $u \in V$), without requiring any explicit knowledge of the $p_v$ vectors.
In fact, we will not even require that there are any vectors $p_v$ for which $\mathcal P^{u,v} = p_u \cdot p_v$.

\paragraph{Euclidean regime}
We say a given Feynman integral computation problem is in the \emph{Euclidean regime} if the matrix $\mathcal P$ is negative semi-definite. 
In this regime, $\cF(\xx) \geq 0$ for all $\xx \in \PP^E_+$.
We call this the Euclidean regime, because the integral \eqref{eq:prop} is equivalent to an analogous Feynman integral where scalar products are computed with the Euclidean all-minus metric. 
To see this, note that as $-\mathcal P$ is positive semi-definite, there is a $|V|\times |V|$ matrix $\mathcal Q$ such that $\mathcal P = - \mathcal Q^T \mathcal Q$.
We can think of the column vectors $\widetilde{p}_1,\ldots,\widetilde{p}_{|V|}$ of $\mathcal Q$ as an auxiliary set of incoming momentum vectors. Elements of $\mathcal P$ can be interpreted as Euclidean, all-minus metric, scalar products of the $\widetilde{p}_v$-vectors:
$\mathcal P^{u,v} = - \widetilde{p}_u^{\, T} \widetilde{p}_v = - \sum_{w \in V} \mathcal Q^{w,u} \mathcal Q^{w,v}$. Translating this back 
to \eqref{eq:prop} means that we can change the signature of the scalar products to the all-minus metric if we replace the external momenta with the $\widetilde{p}_v$ vectors which are defined in an auxiliary space $\RR^{|V|}$ . We emphasize that this way of relating Euclidean and Minkowski space integrals is inherently different from the typical \emph{Wick rotation} procedure and that the $\widetilde p_v$-vectors will in general be different from the original $p_v$ vectors.

\paragraph{Pseudo-Euclidean regime}
In fact, $\cF(\xx) \geq 0$ for all $\xx \in \PP^E_+$ in a larger kinematic regime,
where $\mathcal P$ is not necessarily negative semi-definite.
If for each subset $V' \subset V$ of the vertices the inequality
\begin{align} \label{eq:pQE} \left(\sum_{v \in V'} p_v \right)^2 = \sum_{u,v \in V'} p_u \cdot p_v = \sum_{u,v \in V'} \mathcal P^{u,v} \leq 0 \end{align}
is respected, then we are in  the \emph{pseudo-Euclidean} regime.
The first two equalities in \eqref{eq:pQE} are only included as mnemonic devices;
knowledge of $\mathcal P$ is sufficient to check the inequalities.
Equivalently, we can require the element sums of all \emph{principle minor matrices} of the $\mathcal P$ matrix to be $\leq 0$.

By \eqref{eq:polyUF} and \eqref{eq:pQE}, the coefficients 
of $\cF$ are non-negative in the pseudo-Euclidean regime.
Our choice of normalization factors ensures that \eqref{eq:prop} and \eqref{eq:par} are real positive in this case.

We remark that there is a commonly used alternative definition of a kinematic regime which, on first sight, is similar to the condition above. 
This alternative definition requires the inequalities $p_u \cdot p_v \leq 0$ to be fulfilled for all $u,v \in V$ (see, e.g., \cite[Sec.~2.5]{Weinzierl:2022eaz}). This is more restrictive than our condition in \eqref{eq:pQE}. In fact, it is too restrictive for our purposes, as not even entirely \emph{Euclidean Feynman integrals} can generally be described in this regime. The reason for this is that not all negative semi-definite matrices $\mathcal P$ fulfill this more restrictive condition.

In our case, the Euclidean regime is contained in the 
pseudo-Euclidean regime.
To verify this, we have to make sure that a negative semi-definite $\mathcal P$ fulfills the conditions in \eqref{eq:pQE}. Such a $\mathcal P$  can  be represented with an appropriate set of $\widetilde{p}_v$ vectors as above: $\mathcal P^{u,v} = - \widetilde{p}_u^{\, T} \widetilde{p}_v$.
For each $V' \subset V$ we get the principle minor element sum
\begin{align} \sum_{u,v \in V'}\mathcal P^{u,v} = - \sum_{u,v \in V'} \widetilde{p}_u^{\, T} \widetilde{p}_v = -\left(\sum_{v \in V'} \widetilde{p}_v \right)^T \left(\sum_{v \in V'}\widetilde{p}_v \right) \leq 0 \, . \end{align}

\paragraph{Minkowski regime}
If we are not in the pseudo-Euclidean regime (and thereby also not in the Euclidean regime), %
then we are in the \emph{Minkowski regime}.

\paragraph{Generic and exceptional kinematics}
Without any resort to the explicit incoming momentum vectors $p_v$,
we call a vertex $v$ \emph{internal} if $\mathcal P^{u,v} = 0$ for
all $u \in V$ and \emph{external} otherwise.
Let $V^\textrm{ext} \subset V$ be the set of external vertices.
Complementary to the classification above,
we say that our kinematics are \emph{generic} if
for each \emph{proper} subset $V' \subsetneq V^\textrm{ext}$ of the external vertices of $G$ and
for each non-empty subset $E' \subset E$ of the edges of $G$ we have
\begin{align}\label{eq: generic kinematics}       \left(\sum_{v \in V'} p_v \right)^2 &= \sum_{u,v \in V'} p_u \cdot p_v = \sum_{u,v \in V'} \mathcal P^{u,v} \neq \sum_{e \in E'} m_e^2 \, . \end{align}
For example, the kinematics are always generic 
in the pseudo-Euclidean regime if $m_e > 0$ for all $e \in E$ or if 
$ \sum_{u,v \in V'} \mathcal P^{u,v} < 0$ for all $V' \subsetneq V^\textrm{ext}$.
Note that generic kinematics also exclude 
\emph{on-shell external momenta}, i.e.~cases where $p_v^2 = \mathcal P^{v,v} =0$ for some $v\in V^\textrm{ext}$ 
as long as not all $m_e > 0$, for then there exists at least one edge $e\in E$ such that $p_v^2=0=m_e^2$, thus violating \eqref{eq: generic kinematics}.
Genericity, for instance, guarantees that there will be no cancellation between the momentum and the mass part of the $\cF$-polynomial as defined in \eqref{eq:polyUF}. 

Kinematic configurations that are not generic are called \emph{exceptional}.

As above, only the statements on $\mathcal P^{u,v}$ are sufficient for the classification. The other equalities are added 
to enable a seamless comparison to the literature.

The discussed kinematic regimes and their respective overlaps are illustrated in Figure~\ref{fig:kinematics}. In contrast to what the figure might suggest, the exceptional kinematics only cover a space that is of lower dimension than the one of the generic regime. The Minkowski regime is not explicitly shown as it covers the whole area that is not pseudo-Euclidean. Note that Minkowski, pseudo-Euclidean and Euclidean kinematics can be exceptional.

\ft detects the relevant kinematic regime using the conditions discussed above.

\subsection{Contour deformation}
\label{sec:contour}
In the pseudo-Euclidean (and thereby also in the Euclidean) regime, 
$\mathcal F(\xx)$ stays positive and the integral \eqref{eq:par} cannot have any simple poles inside the integration domain.

In the Minkowski regime however, simple propagator poles of the integrand \eqref{eq:prop} and simple poles associated to zeros of $\mathcal F$ in \eqref{eq:par} are avoided using the causal $i\varepsilon$ prescription (see, e.g., \cite{Eden:1966dnq}). This prescription tells us to which side of the pole the integration contour needs to be deformed. When evaluating integrals such as \eqref{eq:prop} numerically, we have to find an \emph{explicit} choice for such an integration contour. Finding such an explicit contour deformation, which also has decent numerical stability properties, is a surprisingly complicated task. 
Explicit contour deformations for numerical evaluation were pioneered by Soper~\cite{Soper:1999xk} and later refined \cite{Binoth:2005ff,Nagy:2006xy}. This original type of contour deformation has the caveat that the projective symmetry of the integral \eqref{eq:par} is lost as these deformations are inherently non-projective and usually formulated in affine charts, i.e.~`gauge fixed' formulations of \eqref{eq:par}. Experience, e.g.~from \cite{Borinsky:2020rqs}, shows that the projective symmetry of \eqref{eq:par} is a treasured good that should not be given up lightly.

To retain projective symmetry we will hence use a different deformation than established numerical integration tools. 
We will use the embedding $\iota_\lambda: \PP^{E}_+ \hookrightarrow \CP^{|E|-1}$ 
(recall that $\PP^{E}_+$ is a subset of 
$\RP^{|E|-1}$)
of 
the  projective simplex into $|E|-1$  complex dimensional projective space given by
\begin{align} \label{eq:iota} \iota_\lambda : x_e \mapsto x_e \exp \left (-i \lambda\frac{\partial \mathcal V}{\partial x_e}(\xx) \right). \end{align}
This deformation prescription was proposed in \cite[eq.~(43)]{Mizera:2021icv} in the context of the cohomological viewpoint on Feynman integrals (see also \cite[Sec.~4.3]{Hannesdottir:2022bmo}%
). As $\cU$ and $\cF$ are homogeneous polynomials of degree $L$ and $L+1$ respectively and $\cV(\xx) = \cF(\xx)/\cU(\xx)$, the partial derivative $\frac{\partial \mathcal V}{\partial x_e}$ is a rational function in $\xx$ of homogeneous degree $0$, so $\iota_\lambda$ indeed respects projective equivalence.

We want to deform the integration contour $\PP^E_+$ of \eqref{eq:par} into $\iota_\lambda \! \left( \PP^E_+ \right) \subset \CP^{|E|-1}$. 
The deformation $\iota_\lambda$ does not change the boundary of $\PP^E_+$ as each boundary face of $\PP^E_+$ is characterized by at least one vanishing homogeneous coordinate $x_e=0$. So, $\iota_\lambda \! \left( \partial \PP^E_+ \right) = \partial \PP^E_+ $.
By Cauchy's theorem, we can deform the contour as long as 
we do not hit any poles of the integrand $\phi$.
Supposing that $\lambda$ is small enough such that no
poles of $\phi$ are hit by the deformation, we have
\begin{align} \label{eq:deform_formal} \mathcal I = \Gamma(\omega) \int_{\iota_\lambda \left( \PP^{E}_+ \right)} \phi = \Gamma(\omega) \int_{\PP^{E}_+} \iota_\lambda^* \phi \, , \end{align}
where $\iota_\lambda^* \phi$ denotes the \emph{pullback} of the differential form $\phi$. A computation on forms reveals that 
$ \iota_\lambda^* \, \Omega = \det (\mathcal J_\lambda(\xx)) \, \Omega, $
where the Jacobian $\mathcal J_\lambda(\xx)$ is the $|E|\times |E|$ matrix given element-wise by
\begin{align} \label{eq:Jlambda} \mathcal J_\lambda(\xx)^{e,h} =\delta_{e,h} - i \lambda x_e \frac{\partial^2 \mathcal V}{\partial x_e \partial x_h} (\xx) \text{ for all } e,h \in E \, . \end{align}
Thus, we arrive at the desired \emph{deformed parametric Feynman integral} by making \eqref{eq:deform_formal} explicit,
\begin{align} \label{eq:deform} \cI=\Gamma(\omega)\int_{\PP_+^{E}} \iota^*_\lambda \, \phi= \Gamma(\omega)\int_{\PP_+^{E}} \left(\prod_{e\in E}\frac{X_e^{\nu_e} }{\Gamma(\nu_e)}\right) \frac{ \det \mathcal J_\lambda(\xx)   }{\cU\left(\XX\right)^{D/2} \cdot \cV\left(\XX\right)^{\omega}} \, \Omega \, , \end{align}
where $\XX = \iota_\lambda(\xx)$, that means $\XX = (X_0,\ldots,X_{|E|-1})$ and $X_e = x_e \exp\big(-i \lambda \frac{\partial \mathcal V}{\partial x_e} (\xx) \big)$ for all $e\in E$.

Although the prescription \eqref{eq:iota} was proposed before in a more formal context, the deformed formulation of the parametric Feynman integral \eqref{eq:deform} with the explicit Jacobian factor given by \eqref{eq:Jlambda} appears not to have been considered previously in the literature.

In Section~\ref{sec:evaluation}, we provide fast algorithms and formulas to evaluate $\frac{\partial \mathcal V}{\partial x_e} (\xx)$ and $X_e$.

\paragraph{Landau singularities}
In the formulation \eqref{eq:deform}, the $i\varepsilon$ prescription is taken care of by the 
deformation of the rational function $\mathcal V$. To see this, consider the Taylor expansion of $\cV(\XX)$ in $\lambda$,
\begin{align} \cV\left(\XX \right) = \cV(\xx) -i \lambda \sum_{e\in E} x_e \left(\frac{\partial \mathcal V}{\partial x_e} (\xx) \right)^2 + \mathcal O (\lambda^2) \, . \end{align}
The $i\varepsilon$ prescription in \eqref{eq:par} is ensured
if
the imaginary part of $\mathcal V(\XX)$ is strictly negative
for \emph{sufficiently small} $\lambda$.
This is the case for all $\xx \in \PP^E_+$ as long as
there are no solutions of the \emph{Landau equations} %
\begin{align} \label{eq:landau} 0 = x_e \frac{\partial \mathcal V}{\partial x_e} (\xx) \quad \text{ for each } \quad e\in E \, , \quad \text{ for any } \quad \xx \in \PP^E_+ \, , \end{align}
whose solutions are the \emph{Landau singularities}. We will assume that our Feynman integral is always free of Landau singularities. 

Even though we require that $\lambda$ is \emph{small enough}, we can give it, in contrast to the $\varepsilon$ in \eqref{eq:par}, an explicit \emph{finite} value.
Hence, eq.~\eqref{eq:deform} is finally an explicit form of the original Feynman integral \eqref{eq:prop} that is going to serve as input for the tropical numerical integration algorithm.

\subsection{Dimensional regularization and \texorpdfstring{$\epsilon$}{} expansions}
\label{sec:dimreg}

So far, we did not make any restrictions on the finiteness properties of the integrals \eqref{eq:prop}, \eqref{eq:par} and \eqref{eq:deform}.
We say a Feynman integral is \emph{quasi-finite} if the integral in the parametric representation 
\eqref{eq:par} (or equivalently \eqref{eq:deform}) is finite. Only the integral needs to be finite. The 
$\Gamma$ function prefactor is allowed to give divergent contributions.
Note that this is more permissive than requiring that \eqref{eq:prop} is finite, which is already divergent, e.g., for the 
$1$-loop bubble in $D=4$ with unit edge weights. 

In this paper, we will restrict our attention to such quasi-finite Feynman integrals. If an integral is not quasi-finite, it can be expanded as a linear combination of quasi-finite integrals \cite{Nilsson:2013,Berkesch:2014,vonManteuffel:2014qoa}. %

Quasi-finiteness allows \emph{overall} divergences due to the $\Gamma(\omega)$ factor
that becomes singular if $\omega$ is an integer $\leq 0$. Such divergences are easily taken care of by 
using dimensional regularization. As usual we will perturb the dimension by $\epsilon$ in the sense that
\begin{align} D = D_0 - 2\epsilon \, , \end{align}
where $D_0$ is a fixed number and $\epsilon$ is an expansion parameter%
\footnote{
Note that the causal $i\varepsilon$ and the regularization/expansion parameter $\epsilon$ are (unfortunately) usually referred to with the same Greek letter.  We will follow this tradition, but use different versions of the letter for the respective meanings consistently.
}.
Analogously, we define $\omega_0 = \sum_{e\in E} \nu_e - D_0 L/2$.
Using this notation, we may make the $\epsilon$ dependence in \eqref{eq:deform} explicit and expand,
\begin{align} \label{eq:deform_eps} \cI= \Gamma(\omega_0 + \epsilon L) \sum_{k =0}^\infty \frac{ \epsilon^k }{k!} \int_{\PP_+^{E}} \left(\prod_{e\in E}\frac{X_e^{\nu_e} }{\Gamma(\nu_e)}\right) \frac{ \det \mathcal J_\lambda(\xx)   }{\cU\left(\XX\right)^{D_0/2} \cdot \cV\left(\XX\right)^{\omega_0}} \log^k \left( \frac{\cU(\XX)}{\cV(\XX)^L} \right) \, \Omega \, . \end{align}
If the $k=0$ integral is finite, all higher orders in $\epsilon$ are also finite 
as the $\log^k$ factors cannot spoil the integrability. The $\Gamma$ factor can be expanded in $\epsilon$ using
$\Gamma(z+1) = z \Gamma(z)$ and the expansion
\begin{align} \label{eq:gamma_eps} \log \Gamma(1 - \epsilon) = \gamma_E \epsilon + \sum_{n =2}^\infty \frac{\zeta(n)}{n} \epsilon^n \, , \end{align}
with Euler's $\gamma_E$ and Riemann's $\zeta$ function.

Together, eqs.~\eqref{eq:deform_eps} and \eqref{eq:gamma_eps}
give us an explicit formulation of the $\epsilon$ expansion of the Feynman integral \eqref{eq:prop}
in the quasi-finite case. 
In the remainder of this article we will explain how to evaluate the 
expansion coefficients in \eqref{eq:deform_eps} using the tropical 
sampling approach.

\section{Tropical geometry}
\label{sec:tropical}

\subsection{Tropical approximation}
\label{sec:tropicalapprox}

We will use the tropical sampling approach which was put forward in \cite{Borinsky:2020rqs} to evaluate the deformed parametric Feynman integrals in \eqref{eq:deform} and \eqref{eq:deform_eps}.
Here we briefly review the basic concepts.%

For any homogeneous polynomial in $|E|$ variables $p(\xx) = \sum_{k\in \supp(p)} a_k \prod_{e=0}^{|E|-1} x_e^{k_e}$, the support $\supp(p)$ is the set of multi-indices for which $p$ has a non-zero coefficient $a_k$.
For any such polynomial $p$, we define the \emph{tropical approximation} $p^\tr$ as 
\begin{align} p^\tr(\xx)=\max_{k\in\supp(p)} \prod_{e=0}^{|E|-1} x_e^{k_e}\,. \end{align}
If, for example, $p(\xx) = x_0^2 x_1 - 2 x_0 x_1 x_2 + 5 i x_2^3$, then
$p^\tr(\xx) = \max\{ x_0^2 x_1, x_0x_1x_2, x_2^3 \}$.
Note that the tropical approximation forgets about the explicit value of the 
coefficients; it only depends on the fact that a specific coefficient is zero or non-zero.
This way, the tropical approximation only depends on the set $\supp(p)\subset \ZZ_{\geq 0}^{|E|}$. %
In fact, it 
only depends on the shape of the \emph{convex hull} of $\supp(p)$, which is the \emph{Newton polytope} of $p$.
For this reason, $p^\tr$ is nothing but a function avatar of this polytope. 
Indeed, we can write $p^\tr(\xx)$ as follows,
\begin{align} p^\tr(\xx) = \exp \left( \max_{ \mathbf v \in \newt[p] } \mathbf v^T \mathbf y \right), \end{align}
where $\mathbf y = (y_0,\ldots,y_{|E|-1})$ with $y_e = \log x_e$, $\mathbf v^T \mathbf y = \sum_{e\in E} v_e y_e$ and we maximize over the Newton polytope $\newt[p]$ of $p$.
The exponent above is the \emph{tropicalization} $\operatorname{Trop}[p]$ of 
$p$ over $\CC$ with trivial valuation. It plays a central role in tropical geometry (see, e.g., \cite{maclagan2015introduction}).
For us, the key property of the tropical approximation is that it may be used to put upper and lower bounds on a polynomial:
\begin{theorem}[{\cite[Theorem~8]{Borinsky:2020rqs}}]
For a homogeneous $p\in\CC[x_0,\ldots,x_{|E|-1}]$ that is completely non-vanishing on $\PP_+^{E}$ there exist constants $C_1,C_2>0$ such that
\begin{equation}\label{eq: tropical approximation}
    C_1\le\frac{|p(\xx)|}{p^\tr(\xx)}\le C_2 \quad \text{ for all } \quad \xx\in\PP_+^{E} \, .
\end{equation}
\end{theorem}
A polynomial $p$ is \emph{completely non-vanishing} on $\PP^E_+$ if it does not vanish in the interior of $\PP^E_+$ and if another technical condition is fulfilled (see \cite[Definition 1]{Nilsson:2013} for a precise definition).

The $\cU$ polynomial is always completely non-vanishing on $\PP^E_+$ and in the pseudo-Euclidean regime also $\cF$ is completely non-vanishing on $\PP^E_+$. We define the associated tropical approximations $\cU^\tr$, $\cF^\tr$ and $\cV^\tr = \cF^\tr/ \cU^\tr$.

Our key assumption for the integration of Feynman integrals in the Minkowski regime is that the approximation property can also be applied to the deformed Symanzik polynomials.

\begin{assumption}
\label{ass:approx}
There are $\lambda$ dependent constants $C_1(\lambda),C_2(\lambda) > 0$ such that for small $\lambda > 0$,
\begin{align} C_1(\lambda) \leq \left| \left(\frac{\cU^\tr(\xx)}{\cU(\XX)}\right)^{D_0/2} \left(\frac{\cV^\tr(\xx)}{\cV(\XX)}\right)^{\omega_0} \right| \leq C_2(\lambda) \quad \text{ for all } \quad \xx \in \PP^E_+ \, , \end{align}
where we recall that $\XX = (X_1,\ldots,X_{|E|})$ and $X_e = x_e \exp\big(-i \lambda \frac{\partial \mathcal V}{\partial x_e} (\xx) \big)$.
\end{assumption}
In the pseudo-Euclidean regime the assumption is fulfilled, as we are allowed to set $\lambda =0$ and use the established approximation property from \cite{Borinsky:2020rqs} on $\cU$ and $\cF$.
In the Minkowski regime, 
Assumption~\ref{ass:approx} can only be fulfilled if there are no Landau singularities, i.e.~solutions to \eqref{eq:landau}. After extensive numerical testing we conjecture that Assumption~\ref{ass:approx} is fulfilled if there are no Landau singularities.
It would be very interesting to give a concise set of conditions for the validity of 
Assumption~\ref{ass:approx} and how it interplays with such singularities. We leave this to future research.

Another highly promising research question is to find a value for $\lambda$ such that the constants $C_1(\lambda)$ and $C_2(\lambda)$ tighten the bounds as much as possible. Finding such an optimal value for $\lambda$ would result in the first entirely \emph{canonical} deformation prescription which does not depend on free parameters.

\subsection{Tropical sampling}
\label{sec:tropsampling}
Intuitively, Assumption~\ref{ass:approx} tells us that the integrands in \eqref{eq:deform} and \eqref{eq:deform_eps} are, except for phase factors, reasonably approximated by the tropical approximation of the undeformed integrand.
To evaluate the integrals \eqref{eq:deform_eps} with \emph{tropical sampling}, as in \cite[Sec.~7.2]{Borinsky:2020rqs}, we define the probability distribution
\begin{equation}\label{eq:mutr}
    \mu^\tr=\frac{1}{I^\tr}
\frac{\prod_{e\in E}x_e^{\nu_e}}{\cU^\tr(\xx)^{D_0/2}\, \cV^\tr(\xx)^{\omega_0}}\,\Omega \, ,
\end{equation}
where $I^\tr$ is a normalization factor, which is chosen such that $\int_{\PP_+^{E}}\mu^\tr=1$. 
As of Assumption~\ref{ass:approx} and the requirement that 
the integrals in \eqref{eq:deform_eps} shall be finite, 
the factor $I^\tr$ must also be finite. 
If $\omega_0=0$, this normalization factor is equal to the 
associated \emph{Hepp bound} of the graph $G$ \cite{Panzer:2019yxl}.
Because $\mu^\tr >0$ for all $\xx \in \PP^E_+$, 
$\mu^\tr$ gives rise to a proper probability distribution on this domain. 

Using the definition of $\mu^\tr$ to rewrite \eqref{eq:deform_eps} results in
\begin{gather} \begin{gathered} \label{eq:deform_eps_trop} \cI= \frac{ \Gamma(\omega_0 + \epsilon L) }{\prod_{e\in E} \Gamma(\nu_e) } \sum_{k =0}^\infty \frac{ \epsilon^k }{k!} \mathcal I_k, \quad \text{with} \\[6pt] \mathcal I_k = I^\tr \int_{\PP_+^{E}} \frac{ \left(\prod_{e\in E}(X_e/x_e)^{\nu_e} \right) \det \mathcal J_\lambda(\xx)   }{ \left(\cU\left(\XX\right)/\cU^\tr\left(\xx\right)\right)^{D_0/2} \cdot \left(\cV\left(\XX\right)/\cV^\tr\left(\xx\right)\right)^{\omega_0} } \log^k \left( \frac{\cU(\XX)}{\cV(\XX)^L} \right) \mu^\tr \, . \end{gathered} \end{gather}
We will evaluate the integrals above by 
sampling from the probability distribution $\mu^\tr$. 

In \cite{Borinsky:2020rqs}, two different methods to generate samples from $\mu^\tr$ were introduced. The first method \cite[Sec.~5]{Borinsky:2020rqs}, which does not take the explicit structure of $\cU$ and $\cF$ into account, requires the computation of a triangulation of the refined normal fans of the Newton polytopes of $\cU$ and $\cF$. Once such a triangulation is computed, arbitrarily many samples from $\mu^\tr$ can be generated with little computational effort. 
 Unfortunately, obtaining such a triangulation is a highly computationally demanding process. 

The second method \cite[Sec.~6]{Borinsky:2020rqs} to generate samples from the probability distribution $\mu^\tr$ makes use of a particular property of the Newton polytopes of $\cU$ and $\cF$ which allows to bypass the costly triangulation step. This second method additionally has the advantage that it is relatively straightforward to implement. This faster method of sampling from $\mu^\tr$ relies on the Newton polytopes of $\cU$ and $\cF$ being \emph{generalized permutahedra}.

For the program \ft we will make use of this second method.  Our tropical sampling algorithm to produce samples from $\mu^\tr$ is essentially equivalent to the one published with  \cite{Borinsky:2020rqs}. 

\subsection{Base polytopes and generalized permutahedra}
\label{sec:genperm}
A fantastic property of generalized permutahedra is that they come with a \emph{canonical normal fan} which greatly facilitates the sampling of $\mu^\tr$, see \cite[Theorem 27 and Algorithm 4]{Borinsky:2020rqs}. 
Here, we briefly explain the necessary notions. As a start, we define a more general class of polytopes first and discuss restrictions later.

\paragraph{Base polytopes}
Consider a function $z: \mathbf{2}^{E} \rightarrow \RR$ that assigns a number to each subset of $E$, the edge set of our Feynman graph $G$.
In the following we often identify a subset of $E$ with a
subgraph of $G$ and use the respective terms interchangeably. 
So, 
$z$ assigns a number to each subgraph of $G$. We define $\mathbf P[z]$ to be the subset of $\RR^{|E|}$ that  
consists of all points $(a_0,\ldots,a_{|E|-1}) \in \RR^{|E|}$ which fulfill
$\sum_{e\in E} a_e = z(E)$ and the $2^{|E|}-1$ inequalities
\begin{align} \sum_{e \in \gamma} a_e &\geq z(\gamma) \quad \text{ for all } \quad \gamma \subsetneq E\,. \end{align}
Clearly, these inequalities describe a convex bounded domain, i.e.~a \emph{polytope}.
This polytope $\mathbf P[z]$ associated to an arbitrary function $z:\mathbf{2}^E \rightarrow \RR$ is called the \emph{base polytope}. %

\paragraph{Generalized permutahedra}
The following is a special case of a theorem by Aguiar and Ardila who realized that numerous seemingly different structures from combinatorics can be understood using the same object: The \emph{generalized permutahedron} which was initially defined by  Postnikov \cite{Postnikov2009}.
\begin{theorem}[{\cite[Theorem 12.3]{aguiar2017hopf} and the references therein}]
\label{thm:thmgp}
The polytope $\mathbf P[z]$ is a \emph{generalized permutahedron} if and only if the function $z$ is \emph{supermodular}. That means, $z$ fulfills the inequalities
\begin{align} \label{eq:supermod} z(\gamma) + z(\delta) \leq z(\gamma \cup \delta) + z(\gamma \cap \delta) \text{ for all pairs of subgraphs } \gamma, \delta \subset E\,. \end{align}
\end{theorem}

Because other properties of generalized permutahedra are not of central interest in this paper, we will take Theorem~\ref{thm:thmgp} as our definition of these special polytopes. 
Important for us is that for many kinematic situations the Newton polytopes of the Symanzik polynomials are of this type.

Let $L_\gamma$ denote the number of loops of the subgraph $\gamma$, then we have the following theorem due to Schultka \cite{Schultka:2018nrs}:
\begin{theorem}\mbox{}
\label{thm:UgenP}
The Newton polytope $\newt[\cU]$ of $\cU$ is equal to the base polytope $\mathbf P[z_\cU]$ with $z_\cU$ being the function $z_\cU(\gamma) = L_\gamma$. Moreover, $z_\cU$ is supermodular. Hence, by Theorem~\ref{thm:thmgp}, $\newt[\cU]$ is a generalized permutahedron. 
\end{theorem}
\begin{proof}
See \cite[Sec.~4]{Schultka:2018nrs} and the references therein.
In \cite{Panzer:2019yxl}, it was observed that $\newt[\cU]$ is a \emph{matroid polytope}, 
which by \cite[Sec.~14]{aguiar2017hopf} also proves the statement.
\end{proof}

Because $\newt[\cU]$ is a generalized permutahedron, we also say that $\cU$ has the generalized permutahedron property.

\paragraph{Generalized permutahedron property of the $\cF$ polynomial}

For the second Symanzik $\cF$ polynomial the situation is more tricky. 
We need the notion of \emph{mass-momentum spanning} subgraphs which was defined by Brown \cite{Brown:2015fyf} (see also \cite[Sec.~4]{Schultka:2018nrs} for an interesting relationship to the concept of \emph{s-irreducibility} \cite{Chetyrkin:1983wh} or \cite{Speer:1975dc} where related results were obtained or \cite{Beekveldt:2020kzk} for relations to the $R^\star$ operation). We use the following slightly generalized version of Brown's definition (see also \cite[Sec.~7.2]{Borinsky:2020rqs}):
We call a subgraph $\gamma \subset E$ mass-momentum spanning if 
the second Symanzik polynomial of the cograph $G/\gamma$ vanishes identically $\cF_{G/\gamma} = 0$.

\begin{theorem}
\label{thm:FGP}
In the Euclidean regime with generic kinematics, the 
Newton polytope $\newt[\cF]$ is a generalized permutahedron. It is equal to the base polytope $\mathbf P[z_\cF]$ with 
the function $z_\cF$ defined for all subgraphs $\gamma$ by 
$z_\cF(\gamma) = L_{\gamma} + 1$ 
if $\gamma$ is mass-momentum spanning and 
$z_\cF(\gamma) = L_{\gamma}$ otherwise.
 Consequently, this function $z_\cF : \mathbf 2^E \rightarrow \RR$ is supermodular, i.e.\ it fulfills \eqref{eq:supermod}.
\end{theorem}
\begin{proof}
This has also been proven in \cite[Sec.~4]{Schultka:2018nrs}. The proof relies on a special infrared factorization property of $\cF$ that was discovered by Brown \cite[Theorem~2.7]{Brown:2015fyf}.
\end{proof}

We explicitly state the following generalization of Theorem~\ref{thm:FGP}:
\begin{theorem}
\label{thm:generic}
Theorem~\ref{thm:FGP} holds in all regimes
if the kinematics are generic.
\end{theorem}
\begin{proof}  
The $\cF$ polynomial has the same monomials (with different coefficients) as in
the Euclidean regime with generic kinematics. 
To verify this, note that the conditions for generic kinematics 
prevent cancellations between the mass and momentum part of the $\mathcal F$ polynomial
as given in eq.~\eqref{eq:polyUF}.
So, the respective Newton polytopes coincide.
\end{proof}

There is also the following further generalization of Theorem~\ref{thm:FGP} to Euclidean but exceptional kinematics. 
This generalization is very plausible (see \cite[Example~2.5]{Brown:2015fyf}), but it is a technical challenge to prove it. We will not attempt to include a proof here for the sake of brevity. So, we state this generalization as a conjecture:
\begin{conjecture}
\label{conj:exceptEuclid}
Theorem~\ref{thm:FGP} holds in the Euclidean regime 
for all (also exceptional) kinematics.
\end{conjecture}

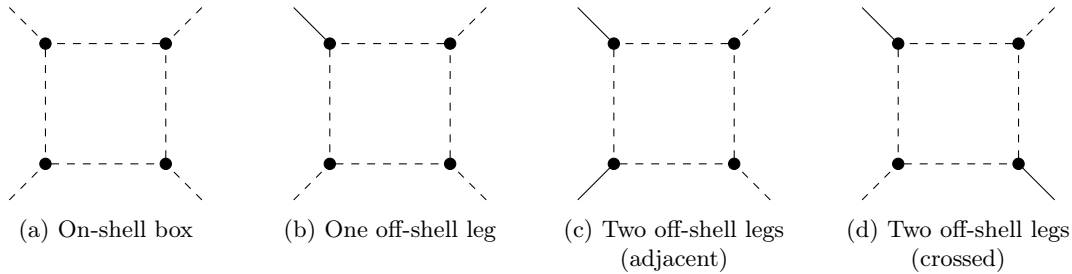
\begin{figure*}%
\captionsetup[subfigure]{justification=centering}
\newcommand{\xs}{.8}%
    \begin{subfigure}[t]{0.25\textwidth}
        \centering
         \begin{tikzpicture}[baseline=-\the\dimexpr\fontdimen22\textfont2\relax] \begin{feynman} \vertex [dot] (v1) at ( \xs, \xs) {}; \vertex [dot] (v2) at (-\xs, \xs) {}; \vertex [dot] (v3) at (-\xs,-\xs) {}; \vertex [dot] (v4) at ( \xs,-\xs) {}; \vertex (i1) at ( {1.6*\xs}, {1.6*\xs}); \vertex (i2) at (-{1.6*\xs}, {1.6*\xs}); \vertex (i3) at (-{1.6*\xs},-{1.6*\xs}); \vertex (i4) at ( {1.6*\xs},-{1.6*\xs}); \diagram*{ (v1)[dot]--[scalar](v2)--[scalar](v3)--[scalar](v4)--[scalar](v1), (i1) -- [scalar](v1), (i2) -- [scalar](v2),(i3) -- [scalar](v3), (i4) -- [scalar] (v4); }; \end{feynman} \end{tikzpicture}
        \caption{On-shell box}
        \label{fig:onshell1}
    \end{subfigure}%
    \begin{subfigure}[t]{0.25\textwidth}
        \centering
         \begin{tikzpicture}[baseline=-\the\dimexpr\fontdimen22\textfont2\relax] \begin{feynman} \vertex [dot] (v1) at ( \xs, \xs) {}; \vertex [dot] (v2) at (-\xs, \xs) {}; \vertex [dot] (v3) at (-\xs,-\xs) {}; \vertex [dot] (v4) at ( \xs,-\xs) {}; \vertex (i1) at ( {1.6*\xs}, {1.6*\xs}); \vertex (i2) at (-{1.6*\xs}, {1.6*\xs}); \vertex (i3) at (-{1.6*\xs},-{1.6*\xs}); \vertex (i4) at ( {1.6*\xs},-{1.6*\xs}); \diagram*{ (v1)[dot]--[scalar](v2)--[scalar](v3)--[scalar](v4)--[scalar](v1), (i1) -- [scalar](v1), (i2) -- (v2),(i3) -- [scalar](v3), (i4) -- [scalar] (v4); }; \end{feynman} \end{tikzpicture}
        \caption{One off-shell leg}
        \label{fig:onshell2}
    \end{subfigure}%
    \begin{subfigure}[t]{0.25\textwidth}
        \centering
        \begin{tikzpicture}[baseline=-\the\dimexpr\fontdimen22\textfont2\relax] \begin{feynman} \vertex [dot] (v1) at ( \xs, \xs) {}; \vertex [dot] (v2) at (-\xs, \xs) {}; \vertex [dot] (v3) at (-\xs,-\xs) {}; \vertex [dot] (v4) at ( \xs,-\xs) {}; \vertex (i1) at ( {1.6*\xs}, {1.6*\xs}); \vertex (i2) at (-{1.6*\xs}, {1.6*\xs}); \vertex (i3) at (-{1.6*\xs},-{1.6*\xs}); \vertex (i4) at ( {1.6*\xs},-{1.6*\xs}); \diagram*{ (v1)[dot]--[scalar](v2)--[scalar](v3)--[scalar](v4)--[scalar](v1), (i1) -- [scalar](v1), (i2) -- (v2),(i3) -- (v3), (i4) -- [scalar] (v4); }; \end{feynman} \end{tikzpicture}
        \caption{Two off-shell legs \\(adjacent)}
        \label{fig:onshell3}
    \end{subfigure}%
    \begin{subfigure}[t]{0.25\textwidth}
        \centering
         \begin{tikzpicture}[baseline=-\the\dimexpr\fontdimen22\textfont2\relax] \begin{feynman} \vertex [dot] (v1) at ( \xs, \xs) {}; \vertex [dot] (v2) at (-\xs, \xs) {}; \vertex [dot] (v3) at (-\xs,-\xs) {}; \vertex [dot] (v4) at ( \xs,-\xs) {}; \vertex (i1) at ( {1.6*\xs}, {1.6*\xs}); \vertex (i2) at (-{1.6*\xs}, {1.6*\xs}); \vertex (i3) at (-{1.6*\xs},-{1.6*\xs}); \vertex (i4) at ( {1.6*\xs},-{1.6*\xs}); \diagram*{ (v1)[dot]--[scalar](v2)--[scalar](v3)--[scalar](v4)--[scalar](v1), (i1) -- [scalar](v1), (i2) -- (v2),(i3) -- [scalar](v3), (i4) -- (v4); }; \end{feynman} \end{tikzpicture}
        \caption{Two off-shell legs \\(crossed)}
        \label{fig:onshell4}
    \end{subfigure}
    \caption{Massless box with different external legs on- or off-shell. On-shell ($p^2=0$) legs are drawn as dashed lines and off-shell ($p^2\neq 0$) legs with solid lines. Internal propagators are massless.}
\label{fig:boxes tikz-feynman dashed}
\end{figure*}

We emphasize that $\newt[\cF]$ is generally \emph{not} a generalized permutahedron outside of the Euclidean regime. This was observed in \cite[Remark~4.16]{Schultka:2018nrs} (see also \cite[Sec.~4.2]{Smirnov:2012gma}, \cite[Sec.~2.2.3]{Panzer:2015ida} or \cite[Remark 35]{Borinsky:2020rqs}).
Explicit counter examples are encountered while computing the massless on-shell boxes depicted in Figure~\ref{fig:boxes tikz-feynman dashed}.
The $\cF$ polynomials of the completely massless box with only on-shell external momenta, the massless box with one off-shell momentum and the massless box with two adjacent off-shell momenta (depicted in Figures~\ref{fig:onshell1}, \ref{fig:onshell2} and \ref{fig:onshell3}) do not fulfill the generalized permutahedron property.
On the other hand, the $\cF$ polynomial does fulfill the generalized permutahedron property for the massless box with two or more off-shell legs such that two off-shell legs are on opposite sides (as depicted in Figure~\ref{fig:onshell4}).

Therefore, we have to make concessions in the Minkowski regime with exceptional kinematics. 

An observation of Arkani-Hamed, Hillman, Mizera is helpful (see 
\cite[eq.~(8)]{Arkani-Hamed:2022cqe} and the discussion around it): the facet presentation of $\newt[\cF]$ given in Theorem~\ref{thm:FGP} turns out to hold in a quite broad range of kinematic regimes, even if $\newt[\cF]$ is not a generalized permutahedron. 
\begin{observation}
\label{obs:GPP}
The Newton polytope of $\cF$ is often equal to the base polytope $\mathbf P[z_\cF]$ with the function $z_\cF$ defined as in Theorem~\ref{thm:FGP}.  %
\end{observation}
This is significant since \ft uses the polytope $\mathbf P[z_\cF]$ internally as a substitute for $\newt[\cF]$ as the former is easier to handle and faster to compute than the latter.

For instance, all massless boxes depicted in Figure~\ref{fig:boxes tikz-feynman dashed} have the property that the Newton polytopes of their $\cF$ polynomial are base polytopes described by the respective $z_\cF$ functions, i.e.~$\newt[\cF] = \mathbf P[z_\cF]$. 
In the first three cases (Figures~\ref{fig:onshell1}, \ref{fig:onshell2}, \ref{fig:onshell3}) the $z_\cF$ function does not fulfill the inequalities \eqref{eq:supermod}. For the graph in Figure~\ref{fig:onshell4} these inequalities are fulfilled and the associated Newton polytope $\newt[\cF] = \mathbf P[z_\cF]$ is a generalized permutahedron.

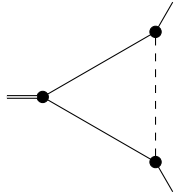
\begin{figure}
\newcommand{\xs}{1}%
        \centering
         \begin{tikzpicture}[baseline=-\the\dimexpr\fontdimen22\textfont2\relax] \begin{feynman} \vertex [dot] (v1) at ( -\xs, 0) {}; \vertex [dot] (v2) at ( {.5*\xs}, {-0.86602540378*\xs}) {}; \vertex [dot] (v3) at ( {.5*\xs}, {0.86602540378*\xs}) {}; \vertex (i1) at ( {-1.6*\xs}, 0) {}; \vertex (i2) at ( {1.6*\xs/2}, {-1.6*\xs*0.86602540378}) {}; \vertex (i3) at ( {1.6*\xs/2}, { 1.6*\xs*0.86602540378}) {}; \diagram*{ (v1)[dot]--(v2)--[scalar](v3)--(v1), (i1) -- [double](v1), (i2) -- (v2),(i3) -- (v3); }; \end{feynman} \end{tikzpicture}
    \caption{Triangle Feynman graph relevant in QED. The two solid propagators have mass $m$ and the solid legs have incoming squared momentum $m^2$. The dashed propagator is massless and the doubled leg has incoming squared momentum $Q^2$.}
    \label{fig:counterexampletriangle}
\end{figure}

It would be very beneficial to have precise conditions for when $\mathbf P[z_\cF]$ indeed is equal to $\newt[\cF]$, we leave this for a future project. Empirically, we have observed that it is valid for quite a wide range of exceptional kinematics. We know, however, that this condition is not fulfilled for arbitrary exceptional kinematics \cite{aaron_email}.
An explicit counter example\footnote{We thank Erik Panzer for sharing this (counter) example with us.} is depicted in Figure~\ref{fig:counterexampletriangle}. For this triangle graph with the indicated exceptional kinematic configuration, the polytope $\newt[\cF]$ is different from $\mathbf P[z_\cF]$. We find that $\cF(\xx) = m^2 (x_1^2 + x_2^2) + (2 m^2 - Q^2 ) x_1 x_2 $ which implies that $\newt[\cF]$ is a one-dimensional polytope. On the other hand, $\mathbf P[z_\cF]$ can be shown to be a two-dimensional polytope. In $D=4$, the Feynman integral associated to Figure~\ref{fig:counterexampletriangle} is infrared divergent and therefore not quasi-finite. In $D=6$, \ft can evaluate the integral without problems. Nonetheless, we expect there to be more complicated Feynman graphs with similarly exceptional external kinematics, that are quasi-finite, but which cannot be evaluated using \soft{feyntrop}. We did not, however, manage to find such a graph.

Even if $\newt[\cF] \neq \mathbf P[z_\cF]$, the Newton polytope $\newt[\cF]$ is \emph{bounded} by the base polytope $\mathbf P[z_\cF]$.
The reason for this is that $\cF$ can only lose monomials if we make the kinematics less generic.
\begin{theorem}
We have $\newt[\cF] \subset \mathbf P[z_\cF]$.
\end{theorem}

\paragraph{Efficient check of the generalized permutahedron property of a base polytope}

Naively, it is quite hard to check if the base polytope $\mathbf{P}[z]$ associated to a given function $z:\mathbf{2}^{E} \rightarrow \RR$ is a generalized permutahedron. There are of the order $2^{2|E|}$ many inequalities to be checked for \eqref{eq:supermod}. 
A more efficient way is to only check the following inequalities
\begin{align} \label{eq:supermod_easy} z(\gamma \cup \{ e \} ) + z(\gamma \cup \{ h \} ) \leq z(\gamma ) + z(\gamma \cup \{ e,h \} ) \end{align}
for all subgraphs $\gamma \subset E$ and edges $e,h \in E\setminus \gamma$.
The inequalities \eqref{eq:supermod_easy} imply the ones in \eqref{eq:supermod}.
For \eqref{eq:supermod_easy} less than $|E|^2 2^{|E|}$ inequalities need to be checked. 
So, \eqref{eq:supermod_easy} is a more efficient version of \eqref{eq:supermod}.

\subsection{Generalized permutahedral tropical sampling}
\label{sec:gentropsampling}

\ft uses a slightly adapted version of the generalized permutahedron tropical sampling algorithm from \cite[Sec.~6.1 and Sec.~7.2]{Borinsky:2020rqs} to sample from the distribution given by $\mu^\tr$ in eq.~\eqref{eq:mutr}.

The algorithm involves a preprocessing and a sampling step.
\paragraph{Preprocessing}
The first algorithmic task to prepare for the sampling from $\mu^\tr$ is to check in which regime the kinematic data are located.
The kinematic data are provided via the matrix $\mathcal P^{u,v}$ as 
it was defined in Section~\ref{sec:def} and via a list of masses $m_e$ for each edge $e\in E$.
If the symmetric $|V|\times |V|$ matrix $\mathcal P^{u,v}$ is negative semi-definite (which is easy to check using matrix diagonalization), then
we are in the Euclidean regime.
Similarly we check if the defining (in)equalities for the other kinematic regimes given in Section~\ref{sec:kinematic_regimes} are fulfilled or not. Depending on the kinematic regime, we need to use a contour deformation for the integration or not. Further, if the kinematics are Euclidean or generic, we know that the generalized permutahedron property of $\cF$ is fulfilled (also thanks to the unproven Conjecture~\ref{conj:exceptEuclid}).
Table~\ref{tab:regime} summarizes this dependence of the algorithm on the kinematic regime.

If we find that we are at an exceptional and non-Euclidean kinematic point, $\mathbf N[\cF]$ might not be a generalized permutahedron 
and it might not even be equal to $\mathbf P[z_\cF]$. In this case, the program prints a message warning the user that the integration might not work. The program then continues under the assumption that $\mathbf N[\cF] = \mathbf P[z_\cF]$.
In any other case, $\newt[\cF]$ is a generalized permutahedron and equal to $\mathbf P[z_\cF]$. Hence, the tropical sampling algorithm is guaranteed to give a convergent Monte Carlo integration method by \cite[Sec.~6.1]{Borinsky:2020rqs}.

The next task is to compute the loop number $L_\gamma$ and check if $\gamma$ is mass-momentum spanning (by asking if $\cF_{G/\gamma} = 0$) for each subgraph $\gamma \subset E$. Using these data, we can compute the values of $z_\cU(\gamma)$ and $z_\cF(\gamma)$ for all subgraphs $\gamma \subset E$
using the respective formulas from Theorems~\ref{thm:UgenP} and \ref{thm:FGP}.

If we are at an exceptional and non-Euclidean kinematic point, we check the inequalities \eqref{eq:supermod_easy} for the $z_\cF$ function. If they are all fulfilled, then $\mathbf P[z_\cF]$ is a generalized permutahedron and we get further indication that the tropical integration step will be successful. The program prints a corresponding message in this case.
Also assuming that Assumption~\ref{ass:approx} is fulfilled, we can compute all integrals in \eqref{eq:deform_eps_trop} efficiently.

Note that even in the pseudo-Euclidean and the Minkowski regimes with exceptional kinematics, the integration is often successful. 
For instance, we can integrate all Feynman graphs depicted in Figure~\ref{fig:boxes tikz-feynman dashed} regardless of the fulfillment of the generalized permutahedron property. In fact, we did not find a quasi-finite example where the algorithm fails (even though the convergence rate is quite bad for examples in highly exceptional kinematic regimes). %
We emphasize, however, that the user should check the convergence of the result separately when integrating at a manifestly exceptional and non-Euclidean kinematic point. For instance, by running the program repeatedly with different numbers of sample points or by slightly perturbing the kinematic point. Recall that for generic kinematics $\newt[\cF]$ is always a generalized permutahedron by Theorem~\ref{thm:generic} and the integration is guaranteed to work if the finiteness assumptions are fulfilled.

\begin{table}[t]
\centering
 \begin{tabular}{l|l|l|l} 
             & Euclidean                 & Pseudo-Euclidean        & Minkowski  \\ [1ex] 
 \hline
 Generic     & no def. / always GP & no def. / always GP     & def. / always GP  \\ \hline
 Exceptional & no def. / always GP & no def. / not always GP & def. / not always GP  \\ [1ex] 
 \end{tabular}
\caption{Table of the necessity of a deformation (def.) and the fulfillment of the generalized permutahedron property of $\cF$ (GP) in each kinematic regime.}
\label{tab:regime}
\end{table}

The next computational step is to compute the 
\emph{generalized degree of divergence} (see \cite[Sec.~7.2]{Borinsky:2020rqs})
for each subgraph $\gamma \subset E$.
It is defined by 
\begin{align} \omega(\gamma) = \sum_{e \in \gamma} \nu_e - D L_\gamma/2 - \omega \delta^{\mathrm{m.m.}}_{\gamma}, \end{align}
where $L_\gamma$ is the loop number of the subgraph $\gamma$ and $\delta^{\mathrm{m.m.}}_{\gamma} = 1$ if $\gamma$ is mass-momentum spanning and $0$ otherwise.
The prefactor $\omega$ of $\delta^{\mathrm{m.m.}}_{\gamma}$ is the usual superficial degree of divergence of the overall graph $G$ as it was defined in Section~\ref{sec:def}, $\omega = \sum_{e\in E} \nu_e - DL/2$.

If $\omega(\gamma) \leq 0$ for any proper subgraph $\gamma$, then 
we discovered a subdivergence. This means that all integrals \eqref{eq:deform_eps} are divergent. Tropical sampling is not possible in this case and the program prints an error message and terminates.
An additional analytic continuation step from \eqref{eq:deform_eps} to a
set of quasi-finite integrals (see~Section~\ref{sec:dimreg}) would resolve this problem. Translating a divergent integral into a linear combination of quasi-finite integrals is always possible, but we will leave the implementation of this step into \ft to a future research project.

If  we have $\omega(\gamma) > 0$ for all $\gamma\subset E$, we
can proceed to the key preparatory step for generalized permutahedral tropical sampling:
We use $\omega(\gamma)$ to compute the following auxiliary subgraph function $J(\gamma)$, which is \emph{recursively} defined by
setting $J(\emptyset) = 1$, agreeing that $\omega(\emptyset)=1$ and
\begin{align} \label{eq:Jrecursion} J(\gamma) = \sum_{e \in \gamma} \frac{J(\gamma \setminus e)}{\omega(\gamma \setminus e)} \text{ for all } \gamma \subset E \, , \end{align}
where $\gamma \setminus e$ is the subgraph $\gamma$ with the edge $e$ removed.
The terminal element of this recursion is the subgraph that contains all edges $E$ of $G$. We find that $J(E) = I^\tr$, where $I^\tr$ is the normalization factor in \eqref{eq:mutr} and \eqref{eq:deform_eps_trop} (see \cite[Proposition~29]{Borinsky:2020rqs} for a proof and details).

In the end of the preprocessing step we compile a table with the information $L_\gamma, \delta^{\mathrm{m.m.}}_\gamma, \omega(\gamma)$ and $J(\gamma)$ for each subgraph $\gamma \subset E$ and store it in the memory of the computer.

\paragraph{Sampling step}
The sampling step of the algorithm repeats the following simple algorithm 
to generate samples $\xx \in \PP^E_+$ that are distributed according to the probability density \eqref{eq:mutr}. 
It is completely described in Algorithm~\ref{alg:gp_sampling}. 
The runtime of our implementation of the  algorithm grows roughly quadratically with $|E|$, but a linear runtime is achievable. %
The validity of the algorithm was proven in a more general setup in \cite[Proposition~31]{Borinsky:2020rqs}.
The additional computation of the values of $\cU^\tr(\xx)$ and $\cV^\tr(\xx)$ is an
application of an optimization algorithm by Fujishige and Tomizawa \cite{fujishige1983note} 
(see also \cite[Lemma~26]{Borinsky:2020rqs}).

The key step of the sampling algorithm is to interpret the recursion \eqref{eq:Jrecursion} as a 
probability distribution for a given subgraph over its edges. That means, for a given 
$\gamma \subset E$ we define $p_e^\gamma = \frac{1}{J(\gamma)} \frac{J(\gamma\setminus e)}{\omega(\gamma\setminus e)}$.
Obviously, $p_e^\gamma \geq 0$ and by \eqref{eq:Jrecursion} we have 
$\sum_{e\in \gamma} p_e^\gamma = 1$.
So, for each $\gamma \subset E$, $p^\gamma_e$ gives a proper probability 
distribution on the edges of the subgraph $\gamma$.

\begin{algorithm}[H]
\begin{algorithmic}[0]
\State{Initialize the variables $\gamma = E$ and $\kappa,U = 1$.}
\While{$\gamma \neq \emptyset$}
\State{Pick a random edge $e \in \gamma$ with probability $p_e^\gamma = \frac{1}{J(\gamma)} \frac{J(\gamma\setminus e)}{\omega(\gamma\setminus e)}$.}
\State{Set $x_e = \kappa$.}
\State{If $\gamma$ is mass-momentum spanning but $\gamma\setminus e$ is not, set $V = x_e$.}
\State{If $L_{\gamma\setminus e} < L_\gamma$, multiply $U$ with $x_e$ and store the result in $U$, i.e.~set $U \gets x_e \cdot U $.}
\State{Remove the edge $e$ from $\gamma$, i.e.~set $\gamma \gets \gamma \setminus e$.}
\State{Pick a uniformly distributed random number $\xi \in [0,1]$.}
\State{Multiply $\kappa$ with $\xi^{1/\omega(\gamma)}$ and store the result in $\kappa$, i.e.~set $\kappa \gets \kappa \xi^{1/\omega(\gamma)}$.}
\EndWhile
\State{Return $\xx = [x_0,\ldots,x_{|E|-1}] \in \PP^E_+$, $\cU^\tr(\xx) = U$ and $\cV^\tr(\xx) = V$.}
\end{algorithmic}
\caption{Generating a sample distributed as $\mu^\tr$ from \eqref{eq:mutr}}
\label{alg:gp_sampling}
\end{algorithm}

The algorithm can also be interpreted as iteratively \emph{cutting} edges of the graph $G$:
We start with $\gamma = E$ and pick a random edge with probability 
$p_e^\gamma$. This edge is \emph{cut} and removed from $\gamma$.
We continue with the newly obtained graph and repeat this cutting process until all edges are removed. In the course of this, Algorithm~\ref{alg:gp_sampling} computes appropriate random values for the coordinates $\xx \in \PP_+^E$.

\section{Numerical integration}

\subsection{Monte Carlo integration}
We now have all the necessary tools at hand to evaluate the 
integrals in  \eqref{eq:deform_eps_trop} using Monte Carlo integration.
In this section, we briefly review this procedure.
The integrals in \eqref{eq:deform_eps_trop} are of the form 
\begin{align} I_f = \int_{\PP_+^E} f(\xx) \mu^\tr \, , \end{align}
where, thanks to the tropical approximation property, $f(\xx)$ is a function that is at most $\log$-singular inside, or on the boundary of, $\PP_+^E$.
To evaluate such an integral, we first use the tropical sampling Algorithm~\ref{alg:gp_sampling} to randomly sample $N$ points $\xx^{(1)}, \ldots, \xx^{(N)} \in \PP_+^E$ that are distributed according to the tropical probability measure $\mu^\tr$.
By the central limit theorem and as $f(\xx)$ is square-integrable,
\begin{align} &I_f \approx I_f^{(N)} & \text{ where } && I_f^{(N)} &= \frac{1}{N} \sum_{i=1}^N f(\xx^{(i)}) \, . \intertext{ For sufficiently large $N$, the expected error of this approximation of the integral $I_f$ is } &\sigma_f = \sqrt{ \frac{ I_{f^2} - I_f^2}{N} } \label{sigma_f} & \text{ where }&& I_{f^2}& = \int_{\PP_+^E} f(\xx)^2 \mu^\tr \, , \intertext{ which itself can be estimated (as long as $f(\xx)^2$ is square-integrable) by } &\sigma_f \approx \sigma_f^{(N)} & \text{ where } & & \sigma_f^{(N)} &= \sqrt{ \frac{1}{N-1} \left( I^{(N)}_{f^2} - \big(I_f^{(N)}\big)^2 \right) } \text{ and } I^{(N)}_{f^2} = \frac{1}{N} \sum_{i=1}^{N} f(\xx^{(i)})^2 \, . \end{align}
To evaluate the estimator $I^{(N)}_f$ and the expected error $\sigma^{(N)}_f$ it is necessary to 
evaluate $f(\xx)$ for $N$ different values of $\xx$. As the random points $\xx^{(1)}, \ldots, \xx^{(N)} \in \PP_+^E$ 
can be obtained quite quickly using  Algorithm~\ref{alg:gp_sampling}, this evaluation becomes a bottleneck. 
In the next section, we describe a fast method to perform this evaluation, which is implemented in \soft{feyntrop} to efficiently obtain Monte Carlo estimates and error terms for the integrals in \eqref{eq:deform_eps_trop}.

\subsection{Fast evaluation of (deformed) Feynman integrands}
\label{sec:evaluation}

To evaluate the integrals in \eqref{eq:deform_eps_trop} using a Monte Carlo 
approach we do not only have to be able to sample from the distribution $\mu^\tr$, but 
we also need to rapidly evaluate the remaining integrand (denoted as $f(\xx)$ in the last section).
Explicitly for the numerical evaluation of \eqref{eq:deform_eps_trop}, we have to be able to  compute $X_e = x_e \exp\big(-i \lambda \frac{\partial \mathcal V}{\partial x_e} (\xx) \big)$ 
as well as 
$\cU(\XX), \cV(\XX)$ and $\det \mathcal J_\lambda(\xx)$ for any  $\xx \in \PP_+^E$.

\paragraph{Evaluation of the $\cU$ and $\cF$ polynomials}

Surprisingly, the explicit polynomial expression for $\cU$ and $\cF$ from eq.~\eqref{eq:polyUF} are \emph{harder} to evaluate than the 
matrix and determinant expression \eqref{eq:laplaceUF} if the underlying graph exceeds a certain complexity. 
The reason for this is that the number of monomials in \eqref{eq:polyUF} increases exponentially with the loop number (see, e.g., \cite{MCKAY1983149} for the asymptotic growth rate of the number of spanning tress in a regular graph),
while the size of the matrices in \eqref{eq:laplaceUF} only increases linearly. Standard linear algebra algorithms as the Cholesky or LU decompositions \cite{horn2012matrix} provide polynomial time algorithms to compute the inverse and determinant of $\mathcal L(\xx)$ and therefore values of $\cU(\xx)$ and $\cF(\xx)$ (see, e.g., \cite[Sec.~7.1]{Borinsky:2020rqs}).
In fact, the linear algebra problems on \emph{graph Laplacian} matrices that need to be solved to compute  
$\cU(\xx)$ and $\cF(\xx)$ fall into a class of problems for which \emph{nearly linear runtime} algorithms are available
\cite{spielman2014nearly}.

\paragraph{Explicit formulas for the $\cV$ derivatives}
We need explicit formulas for the derivatives of $\cV$. These formulas provide fast evaluation methods for $\XX$ and the Jacobian $\mathcal J_\lambda(\xx)$. %

Consider the $(|V|-1) \times (|V|-1)$ matrix $\mathcal M(\xx) = \mathcal L^{-1}(\xx) \,\mathcal P \,\mathcal L^{-1}(\xx)$
with $\mathcal L(\xx)$ and $\mathcal P$ as defined in Section~\ref{sec:def}.
For edges $e$ and $h$ that connect the vertices $u_e,v_e$ and $u_h,v_h$ respectively, we define 
\begin{align} \begin{aligned} \label{eq:AB} \mathcal A(\xx)_{e,h} &= \frac{1}{x_e x_h} \left( \mathcal M(\xx)_{u_e,u_h} + \mathcal M(\xx)_{v_e,v_h} - \mathcal M(\xx)_{u_e,v_h} - \mathcal M(\xx)_{v_e,u_h} \right)\\
\mathcal B(\xx)_{e,h} &= \frac{1}{x_e x_h} \left( \mathcal L^{-1}(\xx)_{u_e,u_h} + \mathcal L^{-1}(\xx)_{v_e,v_h} - \mathcal L^{-1}(\xx)_{u_e,v_h} - \mathcal L^{-1}(\xx)_{v_e,u_h} \right), \end{aligned} \end{align}
where we agree that $\mathcal L^{-1}(\xx)_{u,v} = \mathcal M(\xx)_{u,v} = 0$ if any of $u$ or $v$ is equal to $v_0$, the arbitrary vertex that was removed in the initial expression of the Feynman integral \eqref{eq:prop}.
It follows from \eqref{eq:laplaceUF} and the matrix differentiation rule $\frac{\partial}{\partial x_e } \mathcal L^{-1}(\xx)_{u,v} = \left( -\mathcal L^{-1}(\xx) \frac{\partial \mathcal L}{\partial x_e } (\xx) \mathcal L^{-1}(\xx) \right)_{u,v}$ 
that 
\begin{align} \label{eq:diffV} \frac{\partial \cV}{\partial x_e} (\xx) &= -\mathcal A(\xx)_{e,e} + m_e^2 \, , &       \frac{\partial^2 \cV}{\partial x_e \partial x_h} (\xx) &=  2 \delta_{e,h} \frac{\mathcal A(\xx)_{e,e}}{x_e} - 2 (\mathcal A(\xx) \circ \mathcal B(\xx))_{e,h} \, , \end{align}
where we use the 
 \emph{Hadamard} or \emph{element-wise matrix product},
$(\mathcal A(\xx) \circ \mathcal B(\xx))_{e,h} = \mathcal A(\xx)_{e,h} \cdot \mathcal B(\xx)_{e,h}$.

\paragraph{Computation of the relevant factors in the integrands of \eqref{eq:deform_eps_trop}}
We summarize the necessary steps to compute all the factors in the deformed and $\epsilon$-expanded tropical Feynman integral representation \eqref{eq:deform_eps_trop}.
\begin{enumerate}
\item
Compute the graph Laplacian $\mathcal L(\xx)$ as defined in Section~\ref{sec:def}.
\item
Compute the inverse $\mathcal L^{-1}(\xx)$ (e.g.~by Cholesky decomposing $\mathcal L(\xx)$).
\item
Use this to evaluate the derivatives of $\cV(\xx)$ via 
the formulas in \eqref{eq:AB} and \eqref{eq:diffV}.
\item 
Compute the values of the deformed $\XX$ parameters: $X_e = x_e \exp\big(-i \lambda \frac{\partial \mathcal V}{\partial x_e} (\xx) \big)$.
\item
Compute the Jacobian $\mathcal J_\lambda(\xx)$ using the formula in \eqref{eq:Jlambda}.
\item Evaluate $\det \mathcal J_\lambda(\xx)$ (e.g.~by using a LU decomposition of $\mathcal J_\lambda(\xx)$).
\item Compute the deformed graph Laplacian $\mathcal L(\XX)$.
\item
Compute $\mathcal L^{-1}(\XX)$ and $\det \mathcal L(\XX)$ (e.g.~by using a LU decomposition of $\mathcal L(\XX)$ as a Cholesky decomposition is not possible, because $\mathcal L(\XX)$ is not a hermitian matrix in contrast to $\mathcal L(\xx)$).
\item 
Use the formulas \eqref{eq:laplaceUF}
to obtain values for $\cU(\XX)$, $\cF(\XX)$ and $\cV(\XX) = \cF(\XX)/\cU(\XX)$.
\end{enumerate}

The computation obviously simplifies if 
we set $\lambda =0$, in which case we have $\XX =\xx$. We are allowed to set $\lambda=0$ if we do not need the contour deformation. This is the case, for instance, in the Euclidean or the pseudo-Euclidean regimes. In our implementation we check if we are in these regimes and adjust the evaluation of the integrand accordingly.

\section{The program \ft}
\label{sec:program}
We have implemented the contour-deformed tropical integration algorithm, which we discussed in the previous sections, in a \soft{C++} module named \soft{feyntrop}.
This module is an upgrade to previous code developed by the first author in \cite{Borinsky:2020rqs}.

\ft was checked against \soft{AMFlow} \cite{Liu:2022chg} and \soft{pySecDec} \cite{Borowka:2017idc} for roughly 15 different diagrams with 1-3 loops and 2-5 legs at varying kinematics points, in both the Euclidean and Minkowski regimes, finding agreement in all cases within the given uncertainty bounds.
In the Euclidean regime, the original algorithm was checked against numerous analytic computations that were obtained at high loop order using  conformal four-point integral and graphical function techniques~\cite{Borinsky:2022lds}.

Note that our prefactor convention, which we fixed in eqs.~\eqref{eq:prop} and \eqref{eq:par}, differs from the one in \soft{AMFlow} and \soft{pySecDec} by a factor of $(-1)^{|\mathbf{\nu}|}$, 
where $|\mathbf{\nu}| = \sum_{e=0}^{|E|-1} \nu_e$.
In comparison to \soft{FIESTA} \cite{Smirnov:2021rhf}, our convention differs by a factor of $(-1)^{|\mathbf{\nu}|} \exp\left({-L \gamma_\text{E} \epsilon}\right)$.

\subsection{Installation}

The source code of \ft is available in the repository \href{https://github.com/michibo/feyntrop}{https://github.com/michibo/feyntrop}
on \soft{github}.
It can be downloaded and built by running the following sequence of commands
\begin{verbatim}
git clone https://github.com/michibo/feyntrop.git
cd feyntrop
make
\end{verbatim}
in a \soft{Linux} environment.
\ft is interfaced with \soft{python} \cite{python} via the library \soft{pybind11} \cite{pybind11}%
\footnote{Note added in proof: Due to compatibility issues on some hardware, we removed the dependency on \soft{pybind11} in a new version of \soft{feyntrop} that is available at \href{https://github.com/michibo/feyntrop}{https://github.com/michibo/feyntrop}.  This slightly updated version also includes a low-level command-line interface that does not require \soft{python} at all. This interface might be useful in a high-performance computing environment. It is described in the \texttt{README.md} file in the repository.%
}.
 Additionally, it uses the optimized linear algebra routines from the \soft{Eigen3} package \cite{eigenweb}, the \soft{OpenMP} \soft{C++} module \cite{chandra2001parallel} for the parallelization of the Monte Carlo sampling step and the \soft{xoshiro256+} pseudo random number generator \cite{blackman2021scrambled}.

\ft can be loaded in a \soft{python} environment by importing the file \soft{py\_feyntrop.py}, located in the top directory of the package.
To ensure that \ft was built correctly, one may execute the \soft{python} file \soft{/tests/test\_suite.py}. This script compares the output of \ft against pre-computed values.
To do so, it will locally compute six examples with 1-2 loops and 2-5 legs, some in the Euclidean and others in the Minkowski regime.

The file \soft{py\_feyntrop.py} includes additional functionality for the \soft{python} interface serving three purposes.
Firstly, it simplifies the specification of vertices and edges of a Feynman diagram in comparison to the \soft{C++} interface of \soft{feyntrop}.
Secondly, it allows for self-chosen momentum variables given by a set of replacement rules, instead of having to manually specify the full scalar product matrix $\mathcal P^{u,v}$ from \eqref{eq:laplaceUF}.
Lastly, the output of the $\epsilon$ expansion can be printed in a readable format.

As already indicated in Section~\ref{sec:def}, we employ zero-indexing throughout.
This means that edges and vertices are labeled as $\{0, 1, \ldots\}$. This facilitates 
seamless interoperability with the programming language features of \soft{python}.

\subsection{Basic usage of \texttt{feyntrop}}
In this section, we will illustrate the basic workflow of \ft with an example. The code for this example can be executed and inspected with \soft{jupyter} \cite{Kluyver2016jupyter} by calling
\begin{verbatim}
jupyter notebook tutorial_2L_3pt.ipynb
\end{verbatim}
within the top directory of the \ft package.

We will integrate the following 2-loop 3-point graph in $D = 2 - 2\epsilon$ dimensional spacetime:
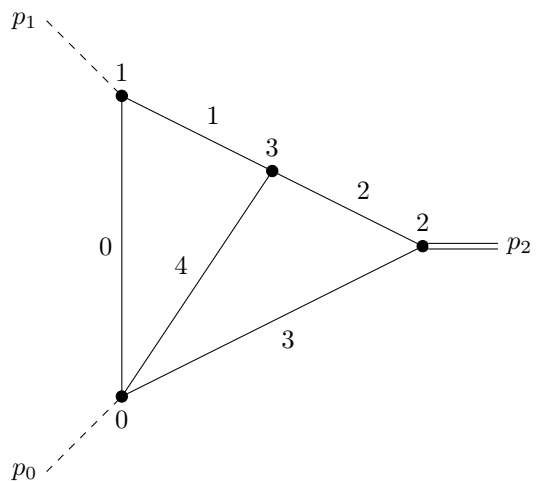
\begin{equation*}
    \newcommand{\xs}{2}
    \centering
    \begin{tikzpicture}[baseline=-\the\dimexpr\fontdimen22\textfont2\relax] \begin{feynman} \vertex [dot,label=\(3\)] (v1) at (0, 0.5*\xs) {}; \vertex [dot,label=\(1\)] (v2) at (-\xs, \xs) {}; \vertex [dot,label=270:\(0\)] (v3) at (-\xs,-\xs) {}; \vertex [dot,label=\(2\)] (v4) at ( \xs,0) {}; \vertex[label=left:{$p_1$}] (i2) at (-{3*\xs/2}, {3*\xs/2}); \vertex[label=left:{$p_0$}] (i3) at (-{3*\xs/2},-{3*\xs/2}); \vertex[label=right:{$p_2$}] (i4) at ( {3*\xs/2},0); \diagram*{ (v1)[dot]--[edge label'=1](v2)--[edge label'=0](v3)--[edge label'=3](v4)--[edge label'=2](v1)--[edge label'=4](v3), (i2) -- [scalar](v2),(i3) -- [scalar](v3), (i4) -- [double,double distance=0.4ex] (v4); }; \end{feynman} \end{tikzpicture} \,.
\end{equation*}
The dashed lines denote on-shell, massless particles with momenta 
$p_0$ and $p_1$ such that $p_0^2 = p_1^2 = 0$. 
The solid, internal lines each have mass $m$.
The double line is associated to some off-shell momentum $p_2^2 \neq 0$.
For the convenience of the reader, both vertices and edges are labeled explicitly in this example.
\ft requires us to label the external vertices (as defined in Section~\ref{sec:kinematic_regimes}) \emph{before} the internal vertices.
In the current example, the vertices are
$ V = V^{\mathrm{ext}} \sqcup V^{\mathrm{int}} = \{0,1,2\} \sqcup \{3\}. $

The momentum space Feynman integral representation \eqref{eq:prop} with unit edge weights reads
\eq{
   \mathcal I =
\pi^{-2+2\epsilon} \, 
   \int 
   \frac
   {\dd^{2-2\epsilon} k_0 \, \dd^{2-2\epsilon} k_1 }
   {
        \big( q_0^2 - m^2 + i \varepsilon \big)
        \big( q_1^2 - m^2 + i \varepsilon\big)
        \big( q_2^2 - m^2 + i \varepsilon\big)
        \big( q_3^2 - m^2+ i \varepsilon \big)
        \big( q_4^2 - m^2 + i \varepsilon\big)
   } \, ,
}
where we integrated out the $\delta$ functions in eq.~\eqref{eq:prop}
by requiring that 
$q_0 = k_0$, $q_1 = k_0 + p_1$, $q_2 = k_0+k_1+p_1$, $q_3 = p_0-k_0-k_1$ and $q_4 = k_1$. 
We choose the phase space point
\eq{
    m^2 = 0.2
    \, , \quad
    p_0^2 = p_1^2 = 0 
    \, , \quad
    p_2^2 = 1 \, ,
}
which is in the Minkowski regime because $p_2^2 > 0$ - see Section~\ref{sec:kinematic_regimes}.
To begin this calculation, first open a \soft{python} script or a \soft{jupyter notebook} and import \soft{py\_feyntrop}:
\begin{verbatim}
from py_feyntrop import *
\end{verbatim}
Here we are assuming that \soft{feyntrop.so} and \soft{py\_feyntrop.py} are both in the working directory.

To define the graph, we provide a list of edges with edge weights $\nu_e$ and squared masses $m_e^2$:
\eq{
   \big( \big(u_0, v_0\big) \, , \, \nu_0 \, , \, m_0^2 \big)
   \, , \, \ldots \, ,
   \big( \big(u_{|E|-1}, v_{|E|-1} \big) \, , \, \nu_{|E|-1} \, , \, m_{|E|-1}^2 \big) \, .
}
The notation $(u_e,v_e)$ denotes an edge $e$ incident to the vertices $u_e$ and $v_e$.
We therefore write
\begin{verbatim}
edges = [((0,1), 1, 'mm'), ((1,3), 1, 'mm'), ((2,3), 1, 'mm'), 
         ((2,0), 1, 'mm'), ((0,3), 1, 'mm')]
\end{verbatim} 
in the code to input the graph which is depicted above. The ordering of vertices $\big(u_e, v_e\big)$ in an edge is insignificant.
Here we set $\nu_e = 1$ for all $e$.
The chosen symbol for $m^2$ is \soft{mm}, which will be replaced by its value $0.2$ later on.
It is also allowed to input numerical values for masses already in the \soft{edges} list, for instance by replacing the first element of the list by \soft{((0,1), 1, '0.2')}.

Next we fix the momentum variables.
Recall that the external vertices are required to be labeled $\{0,1,\ldots,|V^\mathrm{ext}|-1\}$, so the external momenta are $p_0, \ldots, p_{|V^\mathrm{ext}|-1}$.
Moreover, the last momentum is inferred automatically by \ft using momentum conservation, leaving $p_0, \ldots, p_{|V^\mathrm{ext}|-2}$ to be fixed by the user.
A momentum configuration is then specified by the  collection of scalar products,
\eq{
    p_u \cdot p_v \text{  for all } 0 \leq u \leq v\leq |V^\mathrm{ext}|-2 \, .
}
In the code, we must provide replacement rules for these scalar products in terms of some variables of choice.
For the example at hand, $|V^\mathrm{ext}| = 3$, so we must provide replacement rules for $p_0^2, \, p_1^2$ and $p_0 \cdot p_1$.
In the syntax of \ft we thus write
\begin{verbatim}
replacement_rules = [(sp[0,0], '0'), (sp[1,1], '0'), (sp[0,1], 'pp2/2')]
\end{verbatim}
where \soft{sp[u,v]} stands for $p_u \cdot p_v$, the \textbf{s}calar \textbf{p}roduct of $p_u$ and $p_v$.
We have immediately set $p_0^2 = p_1^2 = 0$ and also defined a variable \soft{pp2} which stands for $p_2^2$, as, by momentum conservation, 
\eq{
    p_2^2 = 2 p_0 \cdot p_1 \, .
}

Eventually, we fix numerical values for the two auxiliary variables \soft{pp2} and \soft{mm}.
This is done via
\begin{verbatim}
phase_space_point = [('mm', 0.2), ('pp2', 1)]
\end{verbatim}
which fixes $m^2 = 0.2$ and $p_2^2 = 1$. %
It is possible to obtain the $\mathcal P^{u,v}$ matrix (as defined in Section~\ref{sec:def}) and a list of all the propagator masses, which are computed from the previously provided data, by
\begin{verbatim}
P_uv_matrix, m_sqr_list = prepare_kinematic_data(edges, replacement_rules, 
                                                 phase_space_point) 
\end{verbatim}
The final pieces of data that need to be provided are
\begin{verbatim}
D0 = 2
eps_order = 5
Lambda = 7.6
N = int(1e7)
\end{verbatim}
\soft{D0} is the integer part of the spacetime dimension $D = D_0 - 2\epsilon$.
We expand up to, but not including, \soft{eps\_order}.
\soft{Lambda} denotes the deformation parameter from \eqref{eq:iota}.
\soft{N} is the number of Monte Carlo sampling points.

Tropical Monte Carlo integration of the Feynman integral, with the kinematic configuration chosen above, is now performed by running the command
\begin{verbatim}
trop_res, Itr =  tropical_integration(
                    N, 
                    D0, 
                    Lambda, 
                    eps_order, 
                    edges, 
                    replacement_rules, 
                    phase_space_point)
\end{verbatim}
If the program runs correctly (i.e.\ no error is printed), \soft{trop\_res} will contain the $\epsilon$-expansion \eqref{eq:deform_eps} \emph{without} the prefactor $\Gamma(\omega)/(\Gamma(\nu_1) \cdots \Gamma(\nu_{|E|})) = \Gamma(2 \epsilon + 3)$.
\soft{Itr} is the value of the normalization factor in \eqref{eq:mutr}.
Running this code on a laptop, we get, after a couple of seconds, the output
\begin{verbatim}
Prefactor: gamma(2*eps + 3).
(Effective) kinematic regime: Minkowski (generic).
Generalized permutahedron property: fulfilled.
Analytic continuation: activated. Lambda = 7.6
Started integrating using 8 threads and N = 1e+07 points.
Finished in 6.00369 seconds = 0.00166769 hours.

-- eps^0: [-46.59  +/- 0.13]  +  i * [ 87.19  +/- 0.12]
-- eps^1: [-274.46 +/- 0.55]  +  i * [111.26  +/- 0.55]
-- eps^2: [-435.06 +/- 1.30]  +  i * [-174.47 +/- 1.33]
-- eps^3: [-191.72 +/- 2.15]  +  i * [-494.69 +/- 2.14]
-- eps^4: [219.15  +/- 2.68]  +  i * [-431.96 +/- 2.67]
\end{verbatim}
These printed values for the $\epsilon$ expansion are contained in the list \soft{trop\_res} in the following format:
\begin{align*} \big[ \big( (\text{re}_0, \, \sigma^\text{re}_0) \, , \, (\text{im}_0, \sigma^\text{im}_0) \big) \, , \, \ldots , \big( (\text{re}_4, \sigma^\text{re}_4) \, , \, (\text{im}_4, \sigma^\text{im}_4) \big) \big] \, , \end{align*}
where $\text{re}_0 \pm \sigma^\text{re}_0$ is the real part of the $0$th order term, and so forth.

The $\epsilon$-expansion, with prefactor included, can finally be output via
\begin{verbatim}
eps_expansion(trop_res, edges, D0)
\end{verbatim}
giving
\begin{verbatim}
174.3842115*i - 93.17486662 + eps*(-720.8731714 + 544.3677186*i) + 
eps**2*(-2115.45025 + 496.490128*i) + eps**3*(-3571.990969 - 677.5254794*i) + 
eps**4*(-3872.475723 - 2726.965026*i) + O(eps**5)
\end{verbatim}

If the \soft{tropical\_integration} command fails, for instance because a subdivergence of the input graph is detected, it prints an error message.
The command also prints a warning if the kinematic point is too exceptional and convergence cannot be guaranteed due to the $\cF$ polynomial lacking the generalized permutahedron property (see Section~\ref{sec:genperm}).

\subsection{Deformation parameter}

The uncertainties on the integrated result may greatly vary with the value of the deformation parameter $\lambda$ from \eqref{eq:iota} (what was called \soft{Lambda} above).
Moreover, the optimal value of $\lambda$ might change depending on the phase space point.
It is up to the user to pick a suitable value by trial and error, for instance by integrating several times with a low number of sampling points $N$. In Section~\ref{sec:examples}, this method is used to evaluate multiple examples of Feynman integrals in the Minkowski regime. Typical values for the parameter $\lambda$ can be found there.
It would be beneficial to automate this procedure, possibly by minimizing the sampling variance with respect to~$\lambda$, for instance by solving $\partial_\lambda \sigma_f = 0$ with $\sigma_f$ defined in \eqref{sigma_f}, or by tightening the bounds in Assumption~\ref{ass:approx} (see the discussion after this assumption). We leave the exploration of such ideas to future research.

Note that $\lambda$ has mass dimension $1/\text{mass}^2$. Heuristically, this implies that the value of $\lambda$ should be of order $\mathcal O(1/\Lambda^2)$, where $\Lambda$ is the maximum physical scale in the given computation.

\section{Examples of Feynman integral evaluations}
\label{sec:examples}

In this section, we use \ft to numerically evaluate certain Feynman integrals of interest.
The first two examples, \ref{sec:5L_2pt} and \ref{sec:3L_4pt}, show that \ft is capable of computing Feynman integrals at high loop-orders involving many kinematic scales.
The four examples that follow, \ref{sec:2L_5pt}, \ref{sec:2L_4pt}, \ref{sec:4L_0pt} and  \ref{sec: gg3Higgs}, demonstrate that \ft is capable of computing phenomenologically relevant diagrams.
The final example, \ref{sec:conformal}, is an invitation to study conformal integrals with our code, as they are important for, e.g., $\mathcal N=4$ SYM and the cosmological bootstrap.

We have chosen phase space points which are not close to thresholds to insure good numerical convergence, and expand up to and including $\epsilon^{2L}$ in all but up the last example.

Each of the following examples can be computed with \ft using $10^8$ sampling points within a few minutes on a consumer laptop with 16GBs of RAM. To crosscheck, we used the same machine to evaluate the examples using both \soft{AMFlow}%
\footnote{As \soft{AMFlow} relies on DEQs for Feynman integrals, it is necessary to link it to IBP software. In our examples, we tried the following two options for IBP software: 1) \soft{FIRE} \cite{Smirnov:2019qkx} combined with \soft{LiteRed} \cite{Lee:2012cn,Lee:2013mka}, and 2) \soft{Blade} \cite{blade}.}
and \soft{pySecDec}. 
All computations agreed within the indicated error bounds. Our computations using \soft{AMFlow} and \soft{pySecDec} did not always terminate. Particularly for the Examples \ref{sec:5L_2pt} and \ref{sec:4L_0pt}, neither software finished due to memory constraints of 16GB on our test laptop.
After the initial version of this article became available, Vitaly Magerya informed us that he was able to reproduce also Example~\ref{sec:4L_0pt} and verify our numbers using \soft{pySecDec} with an only slightly more powerful computer. He also found indication that Example~\ref{sec:5L_2pt} is reproducible using a new version of \soft{pySecDec} that was made available three months after the initial version of the present article was posted~\cite{Heinrich:2023til}.

We emphasize that these additional computations using \soft{AMFlow} and \soft{pySecDec} should be seen as a crosscheck and not a benchmark comparison. 
A comparison of \ft and \soft{AMFlow} is difficult as the former directly integrates via Monte Carlo while the latter integrates via differential equations. To integrate a Feynman integral using \soft{AMFlow} an IBP system needs to be solved. Finding this solution is a memory constrained problem and a 16GB laptop is not appropriate to systematically perform computations within this approach. If the IBP system is solved, \soft{AMFlow} provides the evaluated integral at an accuracy which is almost unachievable using a Monte Carlo approach.
The comparison to \soft{pySecDec} is similarly flawed as it can also deal with inherently divergent integrals. To do so it has to check for divergences in each sector which takes time. Moreover, it can deal with completely general algebraic integrals, whereas \ft completely relies on the inherent mathematical structure of Feynman integrals.
We postpone a proper benchmark comparison with the new version of \soft{pySecDec} and updated versions of \soft{AMFlow} to a future research project.

To further highlight the capabilities of \soft{feyntrop}, we computed every example on a high-performance machine, namely a single \soft{AMD EPYC 7H12} 64-core processor using all cores. For each example we use $10^8$ sample points to get a relative accuracy of the order of $10^{-2}$ to $10^{-4}$. The output for each example includes the total evaluation time that \ft needs to compute the respective diagram. This evaluation time includes all steps of the computation. The time needed for the preprocessing step is negligible  in comparison to the sampling time as long as the number of edges is relatively small (i.e.~$|E|\leq 15$). Hence, for such moderate numbers of propagators, the evaluation time is proportional to the number of sample points. The sampling step is completely parallelizable. So, doubling the number of CPUs, halfs the evaluation time. As the evaluation is based on Monte Carlo, increasing the relative accuracy is costly: one additional digit costs a $100$-fold increase in CPU-time. %

The code for each example can be found on the \soft{github} repository in the folder \soft{examples}.

\subsection{A 5-loop 2-point zigzag diagram}
\label{sec:5L_2pt}

We evaluate the following 5-loop 2-point function with all masses different in $D = 3-2\epsilon$ dimensions

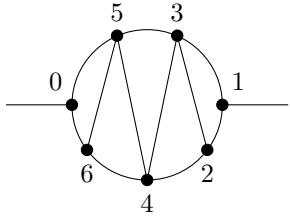
\begin{equation*}
    \centering
    \begin{tikzpicture}[baseline=-\the\dimexpr\fontdimen22\textfont2\relax] \begin{feynman} \vertex [dot,label=93:\(0\)] (v1) at ( 2*0.50,0) {}; \vertex [dot,label=87:\(1\)] (v2) at (2*1.50,0) {}; \vertex [dot,label=270:\(6\)] (v3) at (2*0.60,-2*0.30) {}; \vertex [dot,label=\(5\)] (v4) at ( 2*0.80,2*0.46) {}; \vertex [dot,label=270:\(4\)] (v5) at ( 2*1,-2*0.5) {}; \vertex [dot,label=\(3\)] (v6) at ( 2*1.2,2*0.46) {}; \vertex [dot,label=270:\(2\)] (v7) at ( 2*1.4,-2*0.3) {}; \vertex (a) at (0,0) {}; \vertex (b) at (4,0) {}; \diagram*{ (v3)--(v4)--(v5)--(v6)--(v7),(a)--(v1), (b)--(v2), ; }; \draw({2,0}) circle(1); \end{feynman} \end{tikzpicture}
\end{equation*}
corresponding to the edge set
\begin{verbatim}
edges = [((0,6), 1, '1') , ((0,5), 1, '2'), ((5,6), 1, '3'), 
         ((6,4), 1, '4') , ((5,3), 1, '5'), ((5,4), 1, '6'), 
         ((4,3), 1, '7') , ((4,2), 1, '8'), ((3,2), 1, '9'), 
         ((3,1), 1, '10'), ((2,1), 1, '11')]
\end{verbatim}
Here we already input the chosen values for masses, namely $m_e^2 = e+1$ for $e = 0, \ldots, 10$.

There is only a single independent external momentum $p_0$, whose square we set equal to $100$ via
\begin{verbatim}
replacement_rules = [(sp[0,0], 'pp0')]
phase_space_point = [('pp0', 100)]
\end{verbatim}
The value $\lambda = 0.02$ turns out to give small errors, which is of order $\mathcal O(1/p_0^2)$ in accordance with the comment at the end of the previous section.
Using $N = 10^{8}$ Monte Carlo sampling points, \soft{feyntrop}'s \soft{tropical\_integration} command gives 
{\small
\begin{verbatim}
Prefactor: gamma(5*eps + 7/2).
(Effective) kinematic regime: Minkowski (generic).
Finished in 9.62 seconds.
-- eps^0: [0.0001976 +/- 0.0000016]  +  i * [0.0001415 +/- 0.0000018]
-- eps^1: [-0.004961 +/- 0.000023 ]  +  i * [-0.000802 +/- 0.000024 ]
-- eps^2: [ 0.04943  +/-  0.00017 ]  +  i * [-0.01552  +/-  0.00017 ]
-- eps^3: [-0.25468  +/-  0.00083 ]  +  i * [ 0.24778  +/-  0.00093 ]
-- eps^4: [ 0.5909   +/-  0.0033  ]  +  i * [ -1.7261  +/-  0.0038  ]
-- eps^5: [  1.048   +/-   0.012  ]  +  i * [  7.410   +/-   0.013  ]
-- eps^6: [ -14.652  +/-   0.037  ]  +  i * [ -20.933  +/-   0.038  ]
-- eps^7: [  65.87   +/-   0.10   ]  +  i * [  35.25   +/-   0.11   ]
-- eps^8: [ -190.90  +/-   0.27   ]  +  i * [  -4.91   +/-   0.26   ]
-- eps^9: [ 393.08   +/-   0.70   ]  +  i * [ -182.56  +/-   0.59   ]
-- eps^10:[ -558.01  +/-   1.64   ]  +  i * [ 685.62   +/-   1.29   ]
\end{verbatim}
}
We have not been able to compute this expansion with \soft{AMFlow} for the sake of verification. The memory constraints of 16GB were insufficient.  \soft{pySecDec} applied to this example exhausted the available memory while building the sector decomposition library on our test laptop, but  Vitaly Magerya informed us that he was able to create the integration library on a \texttt{32GB 8-core Intel i7} computer in a couple of hours.
We again emphasize that, for a proper benchmark comparison, our \soft{AMFlow} and \soft{pySecDec} 
code should be put on a machine with more memory. Still, this example illustrates that \ft can operate at high loop order with little memory, CPU and time resources.

\subsection{A 3-loop 4-point envelope diagram}
\label{sec:3L_4pt}

Here, we evaluate a $D = 4-2\epsilon$ dimensional, non-planar, 3-loop 4-point, envelope diagram:

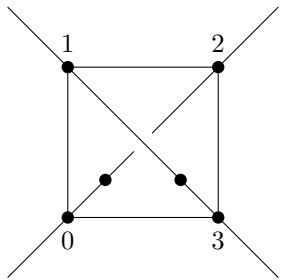
\begin{equation*}%
    \centering
    \begin{tikzpicture}[baseline=-\the\dimexpr\fontdimen22\textfont2\relax] \begin{feynman} \vertex [dot,label={\(2\)}] (v1) at ( 1, 1) {}; \vertex [dot,label={\(1\)}] (v2) at (-1, 1) {}; \vertex [dot,label=270:{\(0\)}] (v3) at (-1,-1) {}; \vertex [dot,label=270:{\(3\)}] (v4) at ( 1,-1) {}; \vertex (i1) at ( {1.8}, {1.8}); \vertex (i2) at (-{1.8}, {1.8}); \vertex (i3) at (-{1.8},-{1.8}); \vertex (i4) at ( {1.8},-{1.8}); \vertex [circle] (o) at (0,0) {}; \vertex [dot] (d1) at ({-0.5},{-0.5}){}; \vertex [dot] (d2) at ({0.5},{-0.5}){}; \diagram*{ (v1)[dot]--(v2)--(v3)--(v4)--(v1), (i1) -- (v1), (i2) -- (v2),(i3) -- (v3), (i4) -- (v4), (v3)--(o)--(v1),(v2)--(v4), }; \end{feynman} \end{tikzpicture}
\end{equation*}
The dots on the crossed lines represent squared propagators, i.e.\ edge weights equal to $2$, rather than vertices.
The weighted edge set with corresponding mass variables is thus
\begin{verbatim}
edges = [((0,1), 1, 'mm0'), ((1,2), 1, 'mm1'), ((2,3), 1, 'mm2'), 
         ((3,0), 1, 'mm3'), ((0,2), 2, 'mm4'), ((1,3), 2, 'mm5')]
\end{verbatim}
Let us define the two-index Mandelstam variables $s_{ij} = (p_i + p_j)^2$, which are put into \soft{feyntrop}'s replacement rules in the form \soft{(sp[i,j], '(sij - ppi - ppj)/2)')} for $0 \leq i < j \leq 2$.
The chosen phase space point is 
\eq{
    &
    p_0^2 = 1.1
    \, , \quad
    p_1^2 = 1.2
    \, , \quad
    p_2^2 = 1.3
    \, , \quad
    s_{01} = 2.1
    \, , \quad
    s_{02} = 2.2
    \, , \quad
    s_{12} = 2.3 \, ,
    \\
    &
    m_0^2 = 0.05
    \,, \quad
    m_1^2 = 0.06
    \,, \quad
    m_2^2 = 0.07
    \,, \quad
    m_3^2 = 0.08
    \, , \quad
    m_4^2 = 0.09
    \, , \quad
    m_5^2 = 0.1 \, .
    \nonumber
}
With additional settings
$ \lambda = 1.24 \text{ and } N = 10^{8} \, , $
we find
{\small
\begin{verbatim}
Prefactor: gamma(3*eps + 2).
(Effective) kinematic regime: Minkowski (generic).
Finished in 5.12 seconds.
-- eps^0: [-10.8335 +/- 0.0084]  +  i * [-12.7145 +/- 0.0083]
-- eps^1: [ 47.971  +/- 0.059 ]  +  i * [-105.057 +/- 0.059 ]
-- eps^2: [ 413.05  +/-  0.23 ]  +  i * [  7.29   +/-  0.23 ]
-- eps^3: [ 372.07  +/-  0.65 ]  +  i * [ 947.82  +/-  0.65 ]
-- eps^4: [-1412.36 +/-  1.45 ]  +  i * [1325.74  +/-  1.45 ]
-- eps^5: [-2726.00 +/-  2.67 ]  +  i * [-1295.36 +/-  2.69 ]
-- eps^6: [ 287.25  +/-  4.28 ]  +  i * [-3982.04 +/-  4.30 ]
\end{verbatim}
}
We verified these numbers using \soft{pySecDec}.
The test machine's memory of 16GBs was exhausted before \soft{AMFlow} could finish the calculation.
The examples in \cite{Liu:2021wks} indicate that using a computer with more memory might also make this 3-loop diagram accessible using \soft{AMFlow}.

\subsection{A 2-loop 4-point \texorpdfstring{$\mu e$}{}-scattering diagram}
\label{sec:2L_4pt}

We evaluate a non-planar, 2-loop 4-point diagram appearing in muon-electron scattering 
\cite{Broggio:2022htr}, which is finite in $D = 6 - 2\epsilon$ dimensions.
It was previously evaluated for vanishing electron mass in \cite{DiVita:2018nnh}.

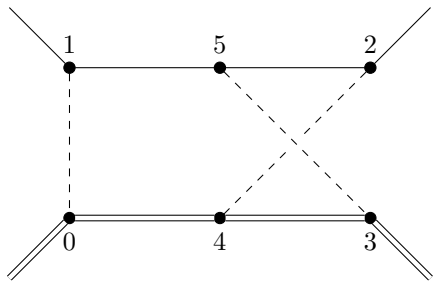
\begin{equation*}
    \centering
    \begin{tikzpicture}[baseline=-\the\dimexpr\fontdimen22\textfont2\relax] \begin{feynman} \vertex [dot,label={\(2\)}] (v1) at ( 1, 1) {}; \vertex [dot,label={\(5\)}] (v2) at (-1, 1) {}; \vertex [dot,label=270:{\(4\)}] (v3) at (-1,-1) {}; \vertex [dot,label=270:{\(3\)}] (v4) at ( 1,-1) {}; \vertex [dot,label={\(1\)}] (v5) at (-3,1) {}; \vertex [dot,label=270:{\(0\)}] (v6) at (-3,-1) {}; \vertex (i1) at ( {1.8}, {1.8}); \vertex (i2) at (-{3.8}, {1.8}); \vertex (i3) at (-{3.8},-{1.8}); \vertex (i4) at ( {1.8},-{1.8}); \vertex [circle] (o) at (0,0) {}; \diagram*{ (i1)--(v1)--(v5)--(i2), (i3)--[double,double distance=0.4ex](v6)--[double,double distance=0.4ex](v3)--[double,double distance=0.4ex](v4)--[double,double distance=0.4ex](i4), (v3)--[scalar](o)--[scalar](v1),(v2)--[scalar](v4), (v5)--[scalar](v6), }; \end{feynman} \end{tikzpicture}
\end{equation*}

The dashed lines represent photons, the solid lines are electrons with mass $m$, and the double lines are muons with mass $M$ (which is approximately 200 times larger than $m$).
The edge set is
\begin{verbatim}
edges = [((0,1), 1, '0'), ((0,4), 1, 'MM'), ((1,5), 1, 'mm'), ((5,2), 1, 'mm'),
         ((5,3), 1, '0'), ((4,3), 1, 'MM'), ((4,2), 1, '0')]
\end{verbatim}
where \soft{MM} and \soft{mm} stand for $M^2$ and $m^2$ respectively.
With a phase space point similar to that of \cite[Section~4.1.2]{DiVita:2018nnh}
\eq{
    p_0^2 &= M^2 = 1
    \, , \quad
    p_1^2 = p_2^2 = m^2 = 1/200
    \, , \quad
    s_{01} = -1/7 \, ,
    \\
    s_{12} &= -1/3
    \, , \quad
    s_{02} = 2M^2 - 2m^2 - s_{01} - s_{12} = 2.49 
    \nonumber
}
and settings
$ \lambda = 1.29 \, , \, N = 10^{8} \, , $
the result becomes
{\small 
\begin{verbatim}
Prefactor: gamma(2*eps + 1).
(Effective) kinematic regime: Minkowski (exceptional).
Finished in 6.53 seconds.
-- eps^0: [1.16483 +/- 0.00083]  +  i * [0.24155 +/- 0.00074]
-- eps^1: [5.5387  +/- 0.0086 ]  +  i * [2.2818  +/- 0.0093 ]
-- eps^2: [15.171  +/-  0.058 ]  +  i * [10.079  +/-  0.064 ]
-- eps^3: [ 28.02  +/-  0.32  ]  +  i * [ 28.17  +/-  0.28  ]
-- eps^4: [ 38.20  +/-  1.42  ]  +  i * [ 56.94  +/-  0.85  ]
\end{verbatim}
}
\noindent 
The momentum configuration is exceptional, so we cannot be sure that the generalized permutahedron property holds - see Section~\ref{sec:genperm}.
In spite of that, \ft gives the correct numbers, which we confirmed using both \soft{AMFlow} and \soft{pySecDec}. 

The leading order term differs from \cite[eq.~(4.20)]{DiVita:2018nnh} by roughly $10\%$ due to our inclusion of the electron mass.
We do, however, reproduce the computation in this reference if we set this mass to $0$ in the \soft{feyntrop} configuration.

\subsection{A QCD-like, 2-loop 5-point diagram}
\label{sec:2L_5pt}

This example is a QCD-like, $D = 6-2\epsilon$ dimensional, 2-loop 5-point diagram:
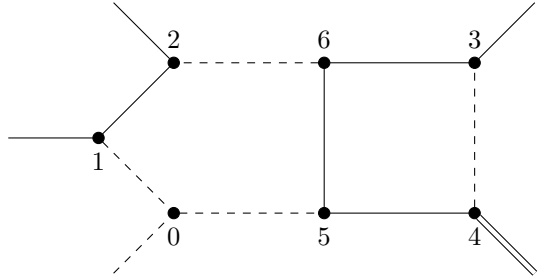
\begin{equation*}
    \centering
    \begin{tikzpicture}[baseline=-\the\dimexpr\fontdimen22\textfont2\relax] \begin{feynman} \vertex [dot,label={\(3\)}] (v1) at ( 1, 1) {}; \vertex [dot,label={\(6\)}] (v2) at (-1, 1) {}; \vertex [dot,label=270:{\(5\)}] (v3) at (-1,-1) {}; \vertex [dot,label=270:{\(4\)}] (v4) at ( 1,-1) {}; \vertex [dot,label={\(2\)}] (v5) at (-3,1) {}; \vertex [dot,label=270:{\(0\)}] (v6) at (-3,-1) {}; \vertex [dot,label=270:{\(1\)}] (v7) at (-4,0) {}; \vertex (i1) at ( {1.8}, {1.8}); \vertex (i2) at (-{3.8}, {1.8}); \vertex (i3) at (-{3.8},-{1.8}); \vertex (i4) at ( {1.8},-{1.8}); \vertex (i5) at (-5.2,0); \vertex [circle] (o) at (0,0) {}; \diagram*{ (i1)--(v1)--(v2)--[scalar](v5)--(v7)--[scalar](v6)--[scalar](v3)--(v4)--[scalar](v1), (v2)--(v3), (i2)--(v5), (i5)--(v7), (i3)--[scalar](v6), (i4)--[double,double distance=0.4ex](v4), }; \end{feynman} \end{tikzpicture}
\end{equation*}
The dashed lines represent gluons, the solid lines are quarks each with mass $m$, and the double line is some off-shell momentum $p_4^2 \neq 0$ fixed by conservation.
The edge data are
\begin{verbatim}
edges = [((0,1), 1, '0'), ((1,2), 1, 'mm'), ((2,6), 1, '0'), ((6,3), 1, 'mm'),
         ((3,4), 1, '0'), ((4,5), 1, 'mm'), ((5,0), 1, '0'), ((5,6), 1, 'mm')]
\end{verbatim}
where \soft{mm} stands for $m^2$.
Let us choose the phase space point
\eq{
    p_0^2 &= 0
    \, , \quad
    p_1^2 = p_2^2 = p_3^2 = m^2 = 1/2
    \, , \quad
    s_{01} = 2.2
    \, , \quad
    s_{02} = 2.3 \, ,
    \\
    s_{03} &= 2.4
    \, , \quad
    s_{12} = 2.5
    \, , \quad
    s_{13} = 2.6
    \, , \quad
    s_{23} = 2.7 \, ,
    \nonumber
}
where again $s_{ij} = (p_i + p_j)^2$.
Finally, setting
$ \lambda = 0.28 \, , \, N = 10^{8} \, , $
we obtain
{\small
\begin{verbatim}
Prefactor: gamma(2*eps + 2).
(Effective) kinematic regime: Minkowski (exceptional).
Finished in 8.20 seconds.
-- eps^0: [0.06480 +/- 0.00078]  +  i * [-0.08150 +/- 0.00098]
-- eps^1: [0.4036  +/- 0.0045 ]  +  i * [ 0.3257  +/- 0.0035 ]
-- eps^2: [-0.7889 +/- 0.0060 ]  +  i * [ 0.957   +/-  0.016 ]
-- eps^3: [-1.373  +/-  0.030 ]  +  i * [ -1.181  +/-  0.034 ]
-- eps^4: [ 1.258  +/-  0.088 ]  +  i * [ -1.205  +/-  0.036 ]
\end{verbatim}
}
\noindent
The kinematic configuration is again exceptional.
Nevertheless, \ft returns the correct numbers, which we verified with \soft{pySecDec}%
\footnote{An earlier version of this article wrongly stated that this computation was not verifiable with \soft{pySecDec}. We thank both an anonymous referee and Vitaly Magerya for pointing this out to us. 
}. We were not able to compute this diagram with \soft{AMFlow} due to our memory constraints.
As similarly intricate Feynman integrals can  be evaluated with \soft{AMFlow} using more memory (see~\cite{Liu:2021wks}), these constraints are very likely the only obstruction for a crosscheck with \soft{AMFlow}.
\subsection{Diagram contributing to triple Higgs production via gluon fusion}
\label{sec: gg3Higgs}
In this example, we evaluate the following diagram contributing to the process%
\footnote{We thank Babis Anastasiou for suggesting this example.}
$gg \to HHH$ in $D = 4 - 2 \epsilon$ dimensions:
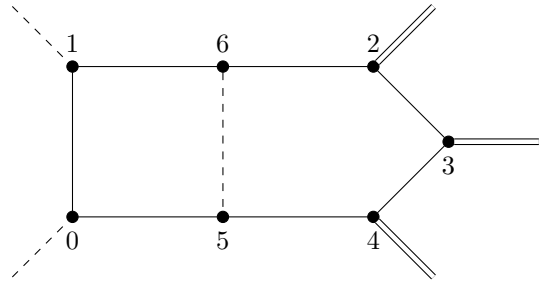
\begin{equation*}
    \centering
    \begin{tikzpicture}[baseline=-\the\dimexpr\fontdimen22\textfont2\relax] \begin{feynman} \vertex [dot,label={\(1\)}] (v1) at ( -1, 1) {}; \vertex [dot,label={\(6\)}] (v2) at (1, 1) {}; \vertex [dot,label=270:{\(5\)}] (v3) at (1,-1) {}; \vertex [dot,label=270:{\(0\)}] (v4) at ( -1,-1) {}; \vertex [dot,label={\(2\)}] (v5) at (3,1) {}; \vertex [dot,label=270:{\(4\)}] (v6) at (3,-1) {}; \vertex [dot,label=270:{\(3\)}] (v7) at (4,0) {}; \vertex (i1) at ( -{1.8}, {1.8}); \vertex (i2) at ({3.8}, {1.8}); \vertex (i3) at ({3.8},-{1.8}); \vertex (i4) at ( -{1.8},-{1.8}); \vertex (i5) at (5.2,0); \vertex [circle] (o) at (0,0) {}; \diagram*{ (i1)--[scalar](v1)--(v2)--(v5)--(v7)--(v6)--(v3)--(v4)--(v1), (v2)--[scalar](v3), (i2)--[double,double distance=0.4ex](v5), (i5)--[double,double distance=0.4ex](v7), (i3)--[double,double distance=0.4ex](v6), (i4)--[scalar](v4), }; \end{feynman} \end{tikzpicture}
\end{equation*}
The dashed lines are massless propagators (representing  gluons), the single solid lines are propagators containing the top quark mass, and the three external double lines are put on-shell to the Higgs mass.
In this case, the list of edges reads
\begin{verbatim}
     edges = [((0,1), 1, 'mm_top'), ((1,6), 1, 'mm_top'), ((5,6), 1, '0'),  
              ((6,2), 1, 'mm_top'), ((2,3), 1, 'mm_top'), ((3,4), 1, 'mm_top'),
              ((4,5), 1, 'mm_top'), ((5,0), 1, 'mm_top')]
\end{verbatim}
with \texttt{mm\_top} being the square of the top quark mass, $m_t^2$.

Given $s_{ij} := (p_i + p_j)^2$, we employ the following kinematic setup:
\eq{
    \nonumber
    p_0^2 &= p_1^2 = 0
    \, , \quad
    p_2^2 = p_3^2 = p_4^2 = m_H^2 
    \, , \\
    s_{01} &= 5 m_H^2 - s_{02} - s_{03} - s_{12} - s_{13} - s_{23} 
    \, .
    \label{eq:gg_3H_s01}
}
The kinematic space is then parameterized by $(s_{02}, \, s_{03}, \, s_{12}, \, s_{13}, \, s_{23}, \, m_t^2, \, m_H^2)$.

Let us evaluate this integral at the phase space point
\eq{
    m_t^2 &= 1.8995
    \, , \quad
    m_H^2 = 1
    \, , \\ \nonumber
    s_{02} = -4.4
    \, , \quad
    s_{03} &= -0.5
    \, , \quad
    s_{12} = -0.6
    \, , \quad
    s_{13} = -0.7
    \, , \quad
    s_{23} = 1.8
    \, ,
}
which lies in the physical region, and has the physically relevant mass ratio $m_t^2 / m_H^2 = 1.8995$.
The remaining Mandelstam invariants are then fixed by momentum conservation to $$(s_{01}, \, s_{04}, \, s_{14}, \, s_{24}, \, s_{34}) = (9.4, \, -1.5, \, -5.1, \, 7.2, \, 3.4).$$

Setting $\lambda = 0.64$ and $N = 10^8$, we get
{\small
\begin{verbatim}
Prefactor: gamma(2*eps + 4).
(Effective) kinematic regime: Minkowski (generic).
Finished in 8.12 seconds.
-- eps^0: [-0.0114757 +/- 0.0000082]  +  i * [0.0035991 +/- 0.0000068]
-- eps^1: [ 0.003250  +/- 0.000031 ]  +  i * [-0.035808 +/- 0.000041 ]
-- eps^2: [ 0.046575  +/- 0.000098 ]  +  i * [0.016143  +/- 0.000088 ]
-- eps^3: [ -0.01637  +/-  0.00017 ]  +  i * [ 0.03969  +/-  0.00016 ]
-- eps^4: [ -0.02831  +/-  0.00023 ]  +  i * [-0.00823  +/-  0.00024 ]
\end{verbatim}
}

We were unable to evaluate this example in reasonable time with \soft{AMFlow}. Again, adding more memory would likely solve this problem. With \soft{pySecDec} we were able to confirm \soft{feyntrop}'s numbers within 3 hours\footnote{
Three months after the initial version of this article was posted, 
a new version of \soft{pySecDec} became available which is, in some cases, up to four times as efficient as the former version~\cite{Heinrich:2023til}.
We postpone a systematic comparison of \ft with this new version to a future research project.
}
 on a laptop, with relative errors around $10^{-2}$. Running \ft on the same laptop with $10^8$ sampling points, we obtain the same numbers within $2.5$ minutes and with relative errors of order $10^{-3}$.

\subsection{A QED-like, 4-loop vacuum diagram}
\label{sec:4L_0pt}

Next we evaluate a QED-like, 4-loop vacuum diagram in $D = 4-2\epsilon$ dimensions:

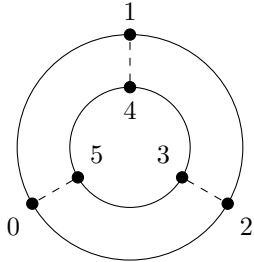
\begin{equation*}
    \centering
    \begin{tikzpicture}[baseline=-\the\dimexpr\fontdimen22\textfont2\relax] \begin{feynman} \vertex [dot,label=270:\(4\)] (v1) at (0,0.8*1) {}; \vertex [dot,label=60:\(5\)] (v2) at ({-0.8*0.866},{-0.8*0.5}) {}; \vertex [dot,label=120:\(3\)] (v3) at ({0.8*0.866},{-0.8*0.5}) {}; \vertex [dot,label=\(1\)] (v4) at (0,1.5*1) {}; \vertex [dot,label=240:\(0\)] (v5) at ({-1.5*0.866},{-1.5*0.5}) {}; \vertex [dot,label=300:\(2\)] (v6) at ({1.5*0.866},{-1.5*0.5}) {}; \diagram*{ (v1)--[scalar](v4),(v2)--[scalar](v5),(v3)--[scalar](v6) ; }; \draw({0,0}) circle(0.8); \draw({0,0}) circle(1.5); \end{feynman} \end{tikzpicture}
\end{equation*}
The dashed lines represent photons, and the solid lines are electrons of mass $m$.
No analytic continuation is required in this case since there are no external momenta - the final result should hence be purely real.
We specify 
\begin{verbatim}
replacement_rules = []
\end{verbatim}
in the code to indicate that all scalar products are zero.

The collection of edges is
\begin{verbatim}
edges = [((0,1), 1, 'mm'), ((1,2), 1, 'mm'), ((2,0), 1, 'mm'),
         ((0,5), 1, '0' ), ((1,4), 1, '0' ), ((2,3), 1, '0' ),
         ((3,4), 1, 'mm'), ((4,5), 1, 'mm'), ((5,3), 1, 'mm')]
\end{verbatim}
where \soft{mm} stands for $m^2$.
Choosing
\begin{verbatim}
phase_space_point = [('mm', 1)]
\end{verbatim}
and setting
$ \lambda = 0 \, , \, N = 10^{8} \, , $
we then find
{\small
\begin{verbatim}
Prefactor: gamma(4*eps + 1).
(Effective) kinematic regime: Euclidean (generic).
Finished in 3.58 seconds.
-- eps^0: [3.01913 +/- 0.00047]  +  i * [0.0 +/- 0.0]
-- eps^1: [-7.0679 +/- 0.0021 ]  +  i * [0.0 +/- 0.0]
-- eps^2: [20.5399 +/- 0.0074 ]  +  i * [0.0 +/- 0.0]
-- eps^3: [-27.895 +/-  0.024 ]  +  i * [0.0 +/- 0.0]
-- eps^4: [62.043  +/-  0.074 ]  +  i * [0.0 +/- 0.0]
-- eps^5: [-59.46  +/-  0.23  ]  +  i * [0.0 +/- 0.0]
-- eps^6: [155.27  +/-  0.73  ]  +  i * [0.0 +/- 0.0]
-- eps^7: [-90.81  +/-  2.26  ]  +  i * [0.0 +/- 0.0]
-- eps^8: [403.78  +/-  6.71  ]  +  i * [0.0 +/- 0.0]
\end{verbatim}
}
We were not able to verify this example with \soft{AMFlow} or \soft{pySecDec} within our memory constraints. 
However,  Vitaly Magerya informed us that he was able to verify these numbers with \soft{pySecDec} in under one hour using an only slightly larger computer.

\subsection{An elliptic, conformal, 4-point integral}
\label{sec:conformal}

The final example is a 1-loop 4-point conformal integral with edge weights $\nu_{1,\ldots,4} = 1/2$ in $D = 2$ dimensions, the result of which was computed in terms of elliptic $K$ functions in \cite[Sec.~7.2]{Corcoran:2021gda}:
\begin{equation}
    \centering
    \begin{tikzpicture}[baseline=-\the\dimexpr\fontdimen22\textfont2\relax] \begin{feynman} \vertex [dot] (o) at (0,0) {}; \vertex [label=\(x_0\)] (v1) at (0,1.5) {}; \vertex [label=\(x_3\)] (v2) at (-1.5,0) {}; \vertex [label=270:\(x_2\)] (v3) at (0,-1.5) {}; \vertex [label=\(x_1\)] (v4) at (1.5,0) {}; \diagram*{ (v1)--(o)--(v2),(v3)--(o)--(v4); }; \end{feynman} \end{tikzpicture}\quad=\frac{4}{\sqrt{-p_2^2}}\left[K(z)K(1-\bar{z})+K(\bar{z})K(1-z)\right]
    \label{conformal_integral}
\end{equation}

\noindent The denominator above differs from \cite[eq. (7.6)]{Corcoran:2021gda} because we have used conformal symmetry to send $x_3 \to \infty$, thereby reducing the kinematic space to that of a 3-point integral.
After identifying dual momentum variables $x_i$ in terms of ordinary momenta as $p_i = x_i - x_{i+1}$, the conformal cross ratios,
with the usual single-valued complex parameterization in terms of $z$ and $\bar z$,
read
\eq{
    z \bar{z} = 
    \frac{p_0^2}{p_2^2} 
    \, , \quad
    (1-z) (1-\bar{z}) = \frac{p_1^2}{p_2^2} \, .
}
In \ft we specify the associated 1-loop 3-point momentum space integral as
\begin{verbatim}
edges = [((0,1), 1/2, '0'), ((1,2), 1/2, '0'), ((2,0), 1/2, '0')]
\end{verbatim}
where all internal masses are zero and edge weights are set to $1/2$.

We choose a momentum configuration in the Euclidean regime:
\eq{
    p_0^2 = -2
    \, , \quad
    p_1^2 = -3
    \, , \quad
    p_2^2 = -5 \, .
}
Although \ft can compute integrals with rational edge weights in the Minkowski regime, it is is most natural to study conformal integrals in the Euclidean regime.

With
$ \lambda = 0 \text{ and } N = 10^8 \, , $
we then obtain
{\small
\begin{verbatim}
(Effective) kinematic regime: Euclidean (generic).
Finished in 1.34 seconds.
-- eps^0: [9.97192 +/- 0.00027]  +  i * [0.0 +/- 0.0]
\end{verbatim}
}
The result agrees with the analytic expression \eqref{conformal_integral}.
This example also illustrates the high efficiency of \ft in the Euclidean regime where 
very high accuracies can be obtained quickly.

\section{Conclusions and outlook}
\label{sec:conclusions}

With this article we introduced \texttt{feyntrop}, a general tool to numerically evaluate quasi-finite Feynman integrals in the physical regime with sufficiently general kinematics. To do so, we gave a  detailed classification of different kinematic regimes that are relevant for numerical integration. Moreover, we presented a completely projective integral expression for concretely $i \varepsilon$-deformed Feynman integrals and their dimensionally regularized expansions.  We used tropical sampling for the numerical integration, which we briefly reviewed, and we discussed the relevant issues on facet presentations of the Newton polytopes of Symanzik polynomials in detail. To be able to perform the numerical integration efficiently, we gave formulas and algorithms for the fast evaluation of Feynman integrals.
To give a concise usage manual for \ft and to illustrate its capabilities, we gave numerous, detailed examples of evaluated Feynman integrals.

The most important restrictions of \ft are 1) it is not capable of dealing with Feynman integrals that have subdivergences (i.e.~non-quasi-finite integrals) and 2) it is not capable of dealing with certain highly exceptional kinematic configurations. 

The first restriction can be lifted by implementing an analytic continuation of the integrand in the spirit of \cite{Nilsson:2013,Berkesch:2014,vonManteuffel:2014qoa} into \soft{feyntrop}. Naively, preprocessing input integrals with such a procedure increases the number of Feynman integrals and thereby also the necessary computer time immensely. However, this proliferation of terms comes from the expansion of the derivatives of the $\cU$ and $\cF$ polynomials as numerators. This expansion can be avoided, because also the derivatives of $\cU$ and $\cF$ (mostly) have the generalized permutahedron property, and because we have fast algorithms to evaluate such derivatives. For instance, we derived a fast algorithm to evaluate the first and second derivatives of $\cF$ in Section~\ref{sec:evaluation}. We postpone the elaboration and implementation of this approach to future work.

A promising approach to lift the second restriction is to try to understand the general shape of the $\cF$ polynomial's Newton polytope. Outside of the Euclidean and generic kinematic regimes, this polytope is not always a generalized permutahedron. In these exceptional kinematic situations, it can have new facets that cannot be explained by known facet presentations. It might be possible to explain these new facets with the help of the Coleman--Norton picture of infrared divergences~\cite{Coleman:1965xm} (see, e.g., \cite{Borinsky:2022msp} where explicit per-diagram factorization of Feynman integrals was observed in a position space based framework).
An alternative approach to fix the issue is to implement the tropical sampling approach that requires a full triangulation of the respective Newton polytopes (see \cite[Sec.~5]{Borinsky:2020rqs}).

Besides this there are numerous, desirable, gradual improvements of \ft that we also postpone to future works. The most important such improvement would be to use the algorithm in conjunction with a \emph{quasi-Monte Carlo} approach. The runtime to obtain the value of an integral up to accuracy $\delta$ currently scales as $\delta^{-2}$, as is standard for a Monte Carlo method. Changing to a quasi-Monte Carlo based procedure would improve this scaling to $\delta^{-1}$. 

Another improvement would be to find an entirely \emph{canonical} deformation prescription. Currently, our deformation still relies on an external parameter that has to be fine-tuned to the respective integral. A canonical deformation prescription that does not depend on a free parameter would lift the burden of this fine-tuning from the user and would likely also produce better rates of convergence.

A more technical update of \ft would involve an implementation of the tropical sampling algorithm on GPUs or on distributed cluster systems. The current implementation of \ft is parallelized and can make use of all cores of a single computer. Running \ft on multiple computers in parallel is not implemented, but there are no technical obstacles to write such an implementation, which we postpone to a future research project.

\section*{Acknowledgements}
We thank Nima Arkani-Hamed, Aaron Hillman, Sebastian Mizera and Erik Panzer for helpful exchanges on facet presentations of Newton polytopes of Symanzik polynomials, Pierpaolo Mastrolia for stimulating discussions on applications to phenomenology,
Yan-Qing Ma for helpful comments on the manuscript
and Vitaly Magerya for comments and independently verifying our numbers in Example \ref{sec:4L_0pt} using \soft{pySecDec}. FT thanks Georgios Papathanasiou for continued support. HJM and FT thank the Institute for Theoretical Studies at the ETH Zürich for hosting the workshop `Tropical and Convex Geometry and Feynman integrals' in August 2022, which was beneficial for the completion of this work. All authors thank the Institute for Advanced Studies, Princeton US, for hospitality during a stay in May 2023 where parts of this work were completed.  MB was supported by Dr.\ Max Rössler, the Walter Haefner Foundation and the ETH Zürich Foundation.  Some of our calculations were carried out on the ETH Euler cluster.

\providecommand{\href}[2]{#2}\begingroup\raggedright\endgroup

\begin{thebibliography}{10}

\bibitem{Heinrich:2020ybq}
G.~Heinrich, \emph{Collider physics at the precision frontier},
  \href{https://doi.org/10.1016/j.physrep.2021.03.006}{\emph{Phys. Rept.}
  {\bfseries 922} (2021) 1} [\href{https://arxiv.org/abs/2009.00516}{{\ttfamily
  2009.00516}}].

\bibitem{Karshenboim:2005iy}
S.G.~Karshenboim, \emph{{Precision physics of simple atoms: QED tests, nuclear
  structure and fundamental constants}},
  \href{https://doi.org/10.1016/j.physrep.2005.08.008}{\emph{Phys. Rept.}
  {\bfseries 422} (2005) 1}
  [\href{https://arxiv.org/abs/hep-ph/0509010}{{\ttfamily hep-ph/0509010}}].

\bibitem{zinn2021quantum}
J.~Zinn-Justin, \emph{Quantum field theory and critical phenomena}, vol.~171,
  Oxford University Press (2021).

\bibitem{Donoghue:1994dn}
J.F.~Donoghue, \emph{{General relativity as an effective field theory: The
  leading quantum corrections}},
  \href{https://doi.org/10.1103/PhysRevD.50.3874}{\emph{Phys. Rev. D}
  {\bfseries 50} (1994) 3874}
  [\href{https://arxiv.org/abs/gr-qc/9405057}{{\ttfamily gr-qc/9405057}}].

\bibitem{Brown:2021umn}
F.~Brown, \emph{Invariant differential forms on complexes of graphs and
  {Feynman} integrals},
  \href{https://doi.org/10.3842/SIGMA.2021.103}{\emph{SIGMA} {\bfseries 17}
  (2021) 103} [\href{https://arxiv.org/abs/2101.04419}{{\ttfamily
  2101.04419}}].

\bibitem{Bern:1994zx}
Z.~Bern, L.J.~Dixon, D.C.~Dunbar and D.A.~Kosower, \emph{{One loop $n$ point
  gauge theory amplitudes, unitarity and collinear limits}},
  \href{https://doi.org/10.1016/0550-3213(94)90179-1}{\emph{Nucl. Phys. B}
  {\bfseries 425} (1994) 217}
  [\href{https://arxiv.org/abs/hep-ph/9403226}{{\ttfamily hep-ph/9403226}}].

\bibitem{Bern:2005cq}
Z.~Bern, L.J.~Dixon and D.A.~Kosower, \emph{{Bootstrapping multi-parton loop
  amplitudes in QCD}},
  \href{https://doi.org/10.1103/PhysRevD.73.065013}{\emph{Phys. Rev. D}
  {\bfseries 73} (2006) 065013}
  [\href{https://arxiv.org/abs/hep-ph/0507005}{{\ttfamily hep-ph/0507005}}].

\bibitem{Chetyrkin:1981qh}
K.G.~Chetyrkin and F.V.~Tkachov, \emph{Integration by parts: The algorithm to
  calculate beta functions in 4 loops},
  \href{https://doi.org/10.1016/0550-3213(81)90199-1}{\emph{Nucl. Phys. B}
  {\bfseries 192} (1981) 159}.

\bibitem{Laporta:2000dsw}
S.~Laporta, \emph{{High precision calculation of multiloop Feynman integrals by
  difference equations}},
  \href{https://doi.org/10.1142/S0217751X00002159}{\emph{Int. J. Mod. Phys. A}
  {\bfseries 15} (2000) 5087}
  [\href{https://arxiv.org/abs/hep-ph/0102033}{{\ttfamily hep-ph/0102033}}].

\bibitem{Bloch:2005bh}
S.~Bloch, H.~Esnault and D.~Kreimer, \emph{{On motives associated to graph
  polynomials}}, \href{https://doi.org/10.1007/s00220-006-0040-2}{\emph{Commun.
  Math. Phys.} {\bfseries 267} (2006) 181}
  [\href{https://arxiv.org/abs/math/0510011}{{\ttfamily math/0510011}}].

\bibitem{Brown:2008um}
F.~Brown, \emph{{The massless higher-loop two-point function}},
  \href{https://doi.org/10.1007/s00220-009-0740-5}{\emph{Commun. Math. Phys.}
  {\bfseries 287} (2009) 925}
  [\href{https://arxiv.org/abs/0804.1660}{{\ttfamily 0804.1660}}].

\bibitem{Panzer:2014caa}
E.~Panzer, \emph{{Algorithms for the symbolic integration of hyperlogarithms
  with applications to Feynman integrals}},
  \href{https://doi.org/10.1016/j.cpc.2014.10.019}{\emph{Comput. Phys. Commun.}
  {\bfseries 188} (2015) 148}
  [\href{https://arxiv.org/abs/1403.3385}{{\ttfamily 1403.3385}}].

\bibitem{Remiddi:1997ny}
E.~Remiddi, \emph{{Differential equations for Feynman graph amplitudes}},
  \href{https://doi.org/10.1007/BF03185566}{\emph{Nuovo Cim. A} {\bfseries 110}
  (1997) 1435} [\href{https://arxiv.org/abs/hep-th/9711188}{{\ttfamily
  hep-th/9711188}}].

\bibitem{Henn:2013pwa}
J.M.~Henn, \emph{{Multiloop integrals in dimensional regularization made
  simple}}, \href{https://doi.org/10.1103/PhysRevLett.110.251601}{\emph{Phys.
  Rev. Lett.} {\bfseries 110} (2013) 251601}
  [\href{https://arxiv.org/abs/1304.1806}{{\ttfamily 1304.1806}}].

\bibitem{Borinsky:2020rqs}
M.~Borinsky, \emph{{Tropical Monte Carlo quadrature for Feynman integrals}},
  \href{https://doi.org/10.4171/AIHPD/158}{\emph{Ann. Inst. Henri Poincar\'e
  Comb. Phys. Interact. (in press)} (2023) }
  [\href{https://arxiv.org/abs/2008.12310}{{\ttfamily 2008.12310}}].

\bibitem{Panzer:2019yxl}
E.~Panzer, \emph{{Hepp's bound for Feynman graphs and matroids}},
  \href{https://doi.org/10.4171/AIHPD/126}{\emph{Ann. Inst. Henri Poincar\'e
  Comb. Phys. Interact.} {\bfseries 10} (2023) 31}
  [\href{https://arxiv.org/abs/1908.09820}{{\ttfamily 1908.09820}}].

\bibitem{Brown:2015fyf}
F.~Brown, \emph{{Feynman amplitudes, coaction principle, and cosmic Galois
  group}}, \href{https://doi.org/10.4310/CNTP.2017.v11.n3.a1}{\emph{Commun.
  Num. Theor. Phys.} {\bfseries 11} (2017) 453}
  [\href{https://arxiv.org/abs/1512.06409}{{\ttfamily 1512.06409}}].

\bibitem{Dunne:2021lie}
G.V.~Dunne and M.~Meynig, \emph{Instantons or renormalons? remarks on
  $\phi_{d=4}^4$ theory in the {MS} scheme},
  \href{https://doi.org/10.1103/PhysRevD.105.025019}{\emph{Phys. Rev. D}
  {\bfseries 105} (2022) 025019}
  [\href{https://arxiv.org/abs/2111.15554}{{\ttfamily 2111.15554}}].

\bibitem{Smirnov:2021rhf}
A.V.~Smirnov, N.D.~Shapurov and L.I.~Vysotsky, \emph{{FIESTA5: Numerical
  high-performance Feynman integral evaluation}},
  \href{https://doi.org/10.1016/j.cpc.2022.108386}{\emph{Comput. Phys. Commun.}
  {\bfseries 277} (2022) 108386}
  [\href{https://arxiv.org/abs/2110.11660}{{\ttfamily 2110.11660}}].

\bibitem{Borinsky:2022}
M.~Borinsky, A.-L.~Sattelberger, B.~Sturmfels and S.~Telen, \emph{Bayesian
  integrals on toric varieties},
  \href{https://doi.org/10.1137/22m1490569}{\emph{SIAM Journal on Applied
  Algebra and Geometry} {\bfseries 7} (2023) 77}
  [\href{https://arxiv.org/abs/2204.06414}{{\ttfamily 2204.06414}}].

\bibitem{Arkani-Hamed:2022cqe}
N.~Arkani-Hamed, A.~Hillman and S.~Mizera, \emph{{Feynman polytopes and the
  tropical geometry of UV and IR divergences}},
  \href{https://doi.org/10.1103/PhysRevD.105.125013}{\emph{Phys. Rev. D}
  {\bfseries 105} (2022) 125013}
  [\href{https://arxiv.org/abs/2202.12296}{{\ttfamily 2202.12296}}].

\bibitem{Arkani-Hamed:2013jha}
N.~Arkani-Hamed and J.~Trnka, \emph{{The amplituhedron}},
  \href{https://doi.org/10.1007/JHEP10(2014)030}{\emph{JHEP} {\bfseries 10}
  (2014) 030} [\href{https://arxiv.org/abs/1312.2007}{{\ttfamily 1312.2007}}].

\bibitem{Arkani-Hamed:2017tmz}
N.~Arkani-Hamed, Y.~Bai and T.~Lam, \emph{Positive geometries and canonical
  forms}, \href{https://doi.org/10.1007/JHEP11(2017)039}{\emph{JHEP} {\bfseries
  11} (2017) 039} [\href{https://arxiv.org/abs/1703.04541}{{\ttfamily
  1703.04541}}].

\bibitem{Nilsson:2013}
L.~Nilsson and M.~Passare, \emph{Mellin transforms of multivariate rational
  functions}, \href{https://doi.org/10.1007/s12220-011-9235-7}{\emph{J. Geom.
  Anal.} {\bfseries 23} (2013) 24}.

\bibitem{Berkesch:2014}
C.~Berkesch, J.~Forsg\aa{}rd and M.~Passare, \emph{Euler-{M}ellin integrals and
  {$A$}-hypergeometric functions},
  \href{https://doi.org/10.1307/mmj/1395234361}{\emph{Michigan Math. J.}
  {\bfseries 63} (2014) 101}.

\bibitem{Schultka:2018nrs}
K.~Schultka, \emph{{Toric geometry and regularization of Feynman integrals}},
  \href{https://arxiv.org/abs/1806.01086}{{\ttfamily 1806.01086}}.

\bibitem{Gelfand1990}
I.M.~Gel'fand, M.M.~Kapranov and A.V.~Zelevinsky, \emph{Generalized {E}uler
  integrals and {$A$}-hypergeometric functions},
  \href{https://doi.org/10.1016/0001-8708(90)90048-R}{\emph{Adv. Math.}
  {\bfseries 84} (1990) 255}.

\bibitem{delaCruz:2019skx}
L.~de~la Cruz, \emph{{Feynman integrals as A-hypergeometric functions}},
  \href{https://doi.org/10.1007/JHEP12(2019)123}{\emph{JHEP} {\bfseries 12}
  (2019) 123} [\href{https://arxiv.org/abs/1907.00507}{{\ttfamily
  1907.00507}}].

\bibitem{Klausen:2019hrg}
R.P.~Klausen, \emph{Hypergeometric series representations of {Feynman}
  integrals by {GKZ} hypergeometric systems},
  \href{https://doi.org/10.1007/JHEP04(2020)121}{\emph{JHEP} {\bfseries 04}
  (2020) 121} [\href{https://arxiv.org/abs/1910.08651}{{\ttfamily
  1910.08651}}].

\bibitem{Klemm:2019dbm}
A.~Klemm, C.~Nega and R.~Safari, \emph{{The $l$-loop Banana Amplitude from GKZ
  Systems and relative Calabi-Yau Periods}},
  \href{https://doi.org/10.1007/JHEP04(2020)088}{\emph{JHEP} {\bfseries 04}
  (2020) 088} [\href{https://arxiv.org/abs/1912.06201}{{\ttfamily
  1912.06201}}].

\bibitem{Chestnov:2022alh}
V.~Chestnov, F.~Gasparotto, M.K.~Mandal, P.~Mastrolia, S.J.~Matsubara-Heo,
  H.J.~Munch et~al., \emph{{Macaulay matrix for Feynman integrals: linear
  relations and intersection numbers}},
  \href{https://doi.org/10.1007/JHEP09(2022)187}{\emph{JHEP} {\bfseries 09}
  (2022) 187} [\href{https://arxiv.org/abs/2204.12983}{{\ttfamily
  2204.12983}}].

\bibitem{Tellander2022}
F.~Tellander and M.~Helmer, \emph{{Cohen-Macaulay Property of Feynman
  Integrals}}, \href{https://doi.org/10.1007/s00220-022-04569-6}{\emph{Commun.
  Math. Phys.} {\bfseries 399} (2023) 1021}
  [\href{https://arxiv.org/abs/2108.01410}{{\ttfamily 2108.01410}}].

\bibitem{Binoth:2000ps}
T.~Binoth and G.~Heinrich, \emph{{An automatized algorithm to compute infrared
  divergent multiloop integrals}},
  \href{https://doi.org/10.1016/S0550-3213(00)00429-6}{\emph{Nucl. Phys. B}
  {\bfseries 585} (2000) 741}
  [\href{https://arxiv.org/abs/hep-ph/0004013}{{\ttfamily hep-ph/0004013}}].

\bibitem{Bogner:2007cr}
C.~Bogner and S.~Weinzierl, \emph{{Resolution of singularities for multi-loop
  integrals}}, \href{https://doi.org/10.1016/j.cpc.2007.11.012}{\emph{Comput.
  Phys. Commun.} {\bfseries 178} (2008) 596}
  [\href{https://arxiv.org/abs/0709.4092}{{\ttfamily 0709.4092}}].

\bibitem{Kaneko:2009qx}
T.~Kaneko and T.~Ueda, \emph{A geometric method of sector decomposition},
  \href{https://doi.org/10.1016/j.cpc.2010.04.001}{\emph{Comput. Phys. Commun.}
  {\bfseries 181} (2010) 1352}
  [\href{https://arxiv.org/abs/0908.2897}{{\ttfamily 0908.2897}}].

\bibitem{Borowka:2017idc}
S.~Borowka, G.~Heinrich, S.~Jahn, S.P.~Jones, M.~Kerner, J.~Schlenk et~al.,
  \emph{{pySecDec: a toolbox for the numerical evaluation of multi-scale
  integrals}}, \href{https://doi.org/10.1016/j.cpc.2017.09.015}{\emph{Comput.
  Phys. Commun.} {\bfseries 222} (2018) 313}
  [\href{https://arxiv.org/abs/1703.09692}{{\ttfamily 1703.09692}}].

\bibitem{Soper:1999xk}
D.E.~Soper, \emph{{Techniques for QCD calculations by numerical integration}},
  \href{https://doi.org/10.1103/PhysRevD.62.014009}{\emph{Phys. Rev. D}
  {\bfseries 62} (2000) 014009}
  [\href{https://arxiv.org/abs/hep-ph/9910292}{{\ttfamily hep-ph/9910292}}].

\bibitem{Anastasiou:2005cb}
C.~Anastasiou and A.~Daleo, \emph{{Numerical evaluation of loop integrals}},
  \href{https://doi.org/10.1088/1126-6708/2006/10/031}{\emph{JHEP} {\bfseries
  10} (2006) 031} [\href{https://arxiv.org/abs/hep-ph/0511176}{{\ttfamily
  hep-ph/0511176}}].

\bibitem{Catani:2008xa}
S.~Catani, T.~Gleisberg, F.~Krauss, G.~Rodrigo and J.-C.~Winter, \emph{{From
  loops to trees by-passing Feynman's theorem}},
  \href{https://doi.org/10.1088/1126-6708/2008/09/065}{\emph{JHEP} {\bfseries
  09} (2008) 065} [\href{https://arxiv.org/abs/0804.3170}{{\ttfamily
  0804.3170}}].

\bibitem{Capatti:2019ypt}
Z.~Capatti, V.~Hirschi, D.~Kermanschah and B.~Ruijl, \emph{Loop-tree duality
  for multiloop numerical integration},
  \href{https://doi.org/10.1103/PhysRevLett.123.151602}{\emph{Phys. Rev. Lett.}
  {\bfseries 123} (2019) 151602}
  [\href{https://arxiv.org/abs/1906.06138}{{\ttfamily 1906.06138}}].

\bibitem{Liu:2017jxz}
X.~Liu, Y.-Q.~Ma and C.-Y.~Wang, \emph{A systematic and efficient method to
  compute multi-loop master integrals},
  \href{https://doi.org/10.1016/j.physletb.2018.02.026}{\emph{Phys. Lett. B}
  {\bfseries 779} (2018) 353}
  [\href{https://arxiv.org/abs/1711.09572}{{\ttfamily 1711.09572}}].

\bibitem{Mandal:2018cdj}
M.K.~Mandal and X.~Zhao, \emph{{Evaluating multi-loop Feynman integrals
  numerically through differential equations}},
  \href{https://doi.org/10.1007/JHEP03(2019)190}{\emph{JHEP} {\bfseries 03}
  (2019) 190} [\href{https://arxiv.org/abs/1812.03060}{{\ttfamily
  1812.03060}}].

\bibitem{Liu:2022chg}
X.~Liu and Y.-Q.~Ma, \emph{{AMFlow: A Mathematica package for Feynman integrals
  computation via auxiliary mass flow}},
  \href{https://doi.org/10.1016/j.cpc.2022.108565}{\emph{Comput. Phys. Commun.}
  {\bfseries 283} (2023) 108565}
  [\href{https://arxiv.org/abs/2201.11669}{{\ttfamily 2201.11669}}].

\bibitem{Hidding:2020ytt}
M.~Hidding, \emph{{DiffExp, a Mathematica package for computing Feynman
  integrals in terms of one-dimensional series expansions}},
  \href{https://doi.org/10.1016/j.cpc.2021.108125}{\emph{Comput. Phys. Commun.}
  {\bfseries 269} (2021) 108125}
  [\href{https://arxiv.org/abs/2006.05510}{{\ttfamily 2006.05510}}].

\bibitem{Armadillo:2022ugh}
T.~Armadillo, R.~Bonciani, S.~Devoto, N.~Rana and A.~Vicini, \emph{{Evaluation
  of Feynman integrals with arbitrary complex masses via series expansions}},
  \href{https://doi.org/10.1016/j.cpc.2022.108545}{\emph{Comput. Phys. Commun.}
  {\bfseries 282} (2023) 108545}
  [\href{https://arxiv.org/abs/2205.03345}{{\ttfamily 2205.03345}}].

\bibitem{Dubovyk:2022frj}
I.~Dubovyk, A.~Freitas, J.~Gluza, K.~Grzanka, M.~Hidding and J.~Usovitsch,
  \emph{{Evaluation of multiloop multiscale Feynman integrals for precision
  physics}}, \href{https://doi.org/10.1103/PhysRevD.106.L111301}{\emph{Phys.
  Rev. D} {\bfseries 106} (2022) L111301}
  [\href{https://arxiv.org/abs/2201.02576}{{\ttfamily 2201.02576}}].

\bibitem{Liu:2022mfb}
Z.-F.~Liu and Y.-Q.~Ma, \emph{Determining {Feynman} integrals with only input
  from linear algebra},
  \href{https://doi.org/10.1103/PhysRevLett.129.222001}{\emph{Phys. Rev. Lett.}
  {\bfseries 129} (2022) 222001}
  [\href{https://arxiv.org/abs/2201.11637}{{\ttfamily 2201.11637}}].

\bibitem{Binoth:2005ff}
T.~Binoth, J.P.~Guillet, G.~Heinrich, E.~Pilon and C.~Schubert, \emph{An
  algebraic/numerical formalism for one-loop multi-leg amplitudes},
  \href{https://doi.org/10.1088/1126-6708/2005/10/015}{\emph{JHEP} {\bfseries
  10} (2005) 015} [\href{https://arxiv.org/abs/hep-ph/0504267}{{\ttfamily
  hep-ph/0504267}}].

\bibitem{Pittau:2021jbs}
R.~Pittau and B.~Webber, \emph{{Direct numerical evaluation of multi-loop
  integrals without contour deformation}},
  \href{https://doi.org/10.1140/epjc/s10052-022-10008-6}{\emph{Eur. Phys. J. C}
  {\bfseries 82} (2022) 55} [\href{https://arxiv.org/abs/2110.12885}{{\ttfamily
  2110.12885}}].

\bibitem{Mizera:2021icv}
S.~Mizera and S.~Telen, \emph{{Landau discriminants}},
  \href{https://doi.org/10.1007/JHEP08(2022)200}{\emph{JHEP} {\bfseries 08}
  (2022) 200} [\href{https://arxiv.org/abs/2109.08036}{{\ttfamily
  2109.08036}}].

\bibitem{vonManteuffel:2014qoa}
A.~von Manteuffel, E.~Panzer and R.M.~Schabinger, \emph{{A quasi-finite basis
  for multi-loop Feynman integrals}},
  \href{https://doi.org/10.1007/JHEP02(2015)120}{\emph{JHEP} {\bfseries 02}
  (2015) 120} [\href{https://arxiv.org/abs/1411.7392}{{\ttfamily 1411.7392}}].

\bibitem{Kol:2007bc}
B.~Kol and M.~Smolkin, \emph{Non-relativistic gravitation: From {Newton to
  Einstein} and back},
  \href{https://doi.org/10.1088/0264-9381/25/14/145011}{\emph{Class. Quant.
  Grav.} {\bfseries 25} (2008) 145011}
  [\href{https://arxiv.org/abs/0712.4116}{{\ttfamily 0712.4116}}].

\bibitem{Nakanishi:1971}
N.~Nakanishi, \emph{{Graph theory and Feynman integrals}}, {Gordon and Breach}
  (1971).

\bibitem{Weinzierl:2022eaz}
S.~Weinzierl, \emph{{Feynman Integrals}}, Springer, Cham (2022),
  \href{https://doi.org/10.1007/978-3-030-99558-4}{10.1007/978-3-030-99558-4},
  [\href{https://arxiv.org/abs/2201.03593}{{\ttfamily 2201.03593}}].

\bibitem{Eden:1966dnq}
R.J.~Eden, P.V.~Landshoff, D.I.~Olive and J.C.~Polkinghorne, \emph{{The
  analytic S-matrix}}, Cambridge University Press (1966).

\bibitem{Nagy:2006xy}
Z.~Nagy and D.E.~Soper, \emph{{Numerical integration of one-loop Feynman
  diagrams for N-photon amplitudes}},
  \href{https://doi.org/10.1103/PhysRevD.74.093006}{\emph{Phys. Rev. D}
  {\bfseries 74} (2006) 093006}
  [\href{https://arxiv.org/abs/hep-ph/0610028}{{\ttfamily hep-ph/0610028}}].

\bibitem{Hannesdottir:2022bmo}
H.S.~Hannesdottir and S.~Mizera, \emph{{What is the i\ensuremath{\varepsilon}
  for the S-matrix?}}, Springer (2023),
  \href{https://doi.org/10.1007/978-3-031-18258-7}{10.1007/978-3-031-18258-7},
  [\href{https://arxiv.org/abs/2204.02988}{{\ttfamily 2204.02988}}].

\bibitem{maclagan2015introduction}
D.~Maclagan and B.~Sturmfels, \emph{Introduction to tropical geometry},
  vol.~161 of \emph{Graduate Studies in Mathematics}, American Mathematical
  Society, Providence, RI (2015),
  \href{https://doi.org/10.1090/gsm/161}{10.1090/gsm/161}.

\bibitem{Postnikov2009}
A.~Postnikov, \emph{Permutohedra, associahedra, and beyond},
  \href{https://doi.org/10.1093/imrn/rnn153}{\emph{Int. Math. Res. Not. IMRN}
  (2009) 1026}.

\bibitem{aguiar2017hopf}
M.~Aguiar and F.~Ardila, \emph{Hopf monoids and generalized permutahedra},
  \href{https://arxiv.org/abs/1709.07504}{{\ttfamily 1709.07504}}.

\bibitem{Chetyrkin:1983wh}
K.G.~Chetyrkin and V.A.~Smirnov, \emph{{Dimensional regularization and infrared
  divergences}}, \href{https://doi.org/10.1007/BF01016818}{\emph{Theor. Math.
  Phys.} {\bfseries 56} (1983) 770}.

\bibitem{Speer:1975dc}
E.~Speer, \emph{Ultraviolet and infrared singularity structure of generic
  {Feynman} amplitudes}, {\emph{Annales de l'I.H.P. Physique th\'eorique}
  {\bfseries 23} (1975) 1}.

\bibitem{Beekveldt:2020kzk}
R.~Beekveldt, M.~Borinsky and F.~Herzog, \emph{{The Hopf algebra structure of
  the $R^{\star}$-operation}},
  \href{https://doi.org/10.1007/JHEP07(2020)061}{\emph{JHEP} {\bfseries 07}
  (2020) 061} [\href{https://arxiv.org/abs/2003.04301}{{\ttfamily
  2003.04301}}].

\bibitem{Smirnov:2012gma}
V.A.~Smirnov, \emph{Analytic tools for {F}eynman integrals}, vol.~250 of
  \emph{Springer Tracts in Modern Physics}, Springer, Heidelberg (2012),
  \href{https://doi.org/10.1007/978-3-642-34886-0}{10.1007/978-3-642-34886-0}.

\bibitem{Panzer:2015ida}
E.~Panzer, \emph{{Feynman integrals and hyperlogarithms}}, Ph.D. thesis,
  Humboldt U., 2015.
\newblock \href{https://arxiv.org/abs/1506.07243}{{\ttfamily 1506.07243}}.
\newblock 10.18452/17157.

\bibitem{aaron_email}
A.~Hillman, S.~Mizera and E.~Panzer. personal communication / January--February
  2023.

\bibitem{fujishige1983note}
S.~Fujishige and N.~Tomizawa, \emph{A note on submodular functions on
  distributive lattices}, \href{https://doi.org/10.15807/jorsj.26.309}{\emph{J.
  Oper. Res. Soc. Japan} {\bfseries 26} (1983) 309}.

\bibitem{MCKAY1983149}
B.D.~McKay, \emph{Spanning trees in regular graphs},
  \href{https://doi.org/10.1016/S0195-6698(83)80045-6}{\emph{Eur. J. Comb.}
  {\bfseries 4} (1983) 149}.

\bibitem{horn2012matrix}
R.A.~Horn and C.R.~Johnson, \emph{Matrix analysis}, Cambridge University Press,
  Cambridge, second~ed. (2013).

\bibitem{spielman2014nearly}
D.A.~Spielman and S.-H.~Teng, \emph{Nearly linear time algorithms for
  preconditioning and solving symmetric, diagonally dominant linear systems},
  \href{https://doi.org/10.1137/090771430}{\emph{SIAM J. Matrix Anal. Appl.}
  {\bfseries 35} (2014) 835}
  [\href{https://arxiv.org/abs/cs/0607105}{{\ttfamily cs/0607105}}].

\bibitem{Borinsky:2022lds}
M.~Borinsky and O.~Schnetz, \emph{{Recursive computation of Feynman periods}},
  \href{https://doi.org/10.1007/JHEP08(2022)291}{\emph{JHEP} {\bfseries 08}
  (2022) 291} [\href{https://arxiv.org/abs/2206.10460}{{\ttfamily
  2206.10460}}].

\bibitem{python}
G.~Van~Rossum and F.L.~Drake, \emph{Python 3 Reference Manual}, CreateSpace,
  Scotts Valley, CA (2009).

\bibitem{pybind11}
W.~Jakob, J.~Rhinelander and D.~Moldovan, ``pybind11 -- seamless operability
  between {C++11} and python.'' \url{https://github.com/pybind/pybind11}, 2017.

\bibitem{eigenweb}
G.~Guennebaud, B.~Jacob et~al., ``Eigen v3.'' \url{http://eigen.tuxfamily.org},
  2010.

\bibitem{chandra2001parallel}
R.~Chandra, L.~Dagum, D.~Kohr, R.~Menon, D.~Maydan and J.~McDonald,
  \emph{Parallel programming in OpenMP}, Morgan Kaufmann (2001).

\bibitem{blackman2021scrambled}
D.~Blackman and S.~Vigna, \emph{Scrambled linear pseudorandom number
  generators}, \href{https://doi.org/10.1145/3460772}{\emph{ACM Trans. Math.
  Software} {\bfseries 47} (2021) Art. 36, 32}.

\bibitem{Kluyver2016jupyter}
T.~Kluyver, B.~Ragan-Kelley, F.~P{\'e}rez, B.~Granger, M.~Bussonnier,
  J.~Frederic et~al., \emph{Jupyter notebooks -- a publishing format for
  reproducible computational workflows},  in \emph{Positioning and Power in
  Academic Publishing: Players, Agents and Agendas}, F.~Loizides and
  B.~Schmidt, eds., pp.~87 -- 90, IOS Press, 2016.

\bibitem{Smirnov:2019qkx}
A.V.~Smirnov and F.S.~Chuharev, \emph{{FIRE6: Feynman Integral REduction with
  Modular Arithmetic}},
  \href{https://doi.org/10.1016/j.cpc.2019.106877}{\emph{Comput. Phys. Commun.}
  {\bfseries 247} (2020) 106877}
  [\href{https://arxiv.org/abs/1901.07808}{{\ttfamily 1901.07808}}].

\bibitem{Lee:2012cn}
R.N.~Lee, \emph{{Presenting LiteRed: a tool for the Loop InTEgrals REDuction}},
   \href{https://arxiv.org/abs/1212.2685}{{\ttfamily 1212.2685}}.

\bibitem{Lee:2013mka}
R.N.~Lee, \emph{{LiteRed 1.4: a powerful tool for reduction of multiloop
  integrals}}, \href{https://doi.org/10.1088/1742-6596/523/1/012059}{\emph{J.
  Phys. Conf. Ser.} {\bfseries 523} (2014) 012059}
  [\href{https://arxiv.org/abs/1310.1145}{{\ttfamily 1310.1145}}].

\bibitem{blade}
X.~Guan, X.~Liu, Y.-Q.~Ma and W.-H.~Wu, ``{Blade: A package for
  block-triangular form improved Feynman integrals decomposition}.''
  \url{https://gitlab.com/multiloop-pku/blade/}, Accessed: May 17, 2023.

\bibitem{Heinrich:2023til}
G.~Heinrich, S.P.~Jones, M.~Kerner, V.~Magerya, A.~Olsson and J.~Schlenk,
  \emph{Numerical scattering amplitudes with {pySecDec}},
  \href{https://arxiv.org/abs/2305.19768}{{\ttfamily 2305.19768}}.

\bibitem{Liu:2021wks}
X.~Liu and Y.-Q.~Ma, \emph{{Multiloop corrections for collider processes using
  auxiliary mass flow}},
  \href{https://doi.org/10.1103/PhysRevD.105.L051503}{\emph{Phys. Rev. D}
  {\bfseries 105} (2022) L051503}
  [\href{https://arxiv.org/abs/2107.01864}{{\ttfamily 2107.01864}}].

\bibitem{Broggio:2022htr}
A.~Broggio et~al., \emph{{Muon-electron scattering at NNLO}},
  \href{https://doi.org/10.1007/JHEP01(2023)112}{\emph{JHEP} {\bfseries 01}
  (2023) 112} [\href{https://arxiv.org/abs/2212.06481}{{\ttfamily
  2212.06481}}].

\bibitem{DiVita:2018nnh}
S.~Di~Vita, S.~Laporta, P.~Mastrolia, A.~Primo and U.~Schubert, \emph{{Master
  integrals for the NNLO virtual corrections to $\mu e$ scattering in QED: the
  non-planar graphs}},
  \href{https://doi.org/10.1007/JHEP09(2018)016}{\emph{JHEP} {\bfseries 09}
  (2018) 016} [\href{https://arxiv.org/abs/1806.08241}{{\ttfamily
  1806.08241}}].

\bibitem{Corcoran:2021gda}
L.~Corcoran, F.~Loebbert and J.~Miczajka, \emph{{Yangian Ward identities for
  fishnet four-point integrals}},
  \href{https://doi.org/10.1007/JHEP04(2022)131}{\emph{JHEP} {\bfseries 04}
  (2022) 131} [\href{https://arxiv.org/abs/2112.06928}{{\ttfamily
  2112.06928}}].

\bibitem{Coleman:1965xm}
S.~Coleman and R.E.~Norton, \emph{{Singularities in the physical region}},
  \href{https://doi.org/10.1007/BF02750472}{\emph{Nuovo Cim.} {\bfseries 38}
  (1965) 438}.

\bibitem{Borinsky:2022msp}
M.~Borinsky, Z.~Capatti, E.~Laenen and A.~Salas-Bern\'ardez,
  \emph{{Flow-oriented perturbation theory}},
  \href{https://doi.org/10.1007/JHEP01(2023)172}{\emph{JHEP} {\bfseries 01}
  (2023) 172} [\href{https://arxiv.org/abs/2210.05532}{{\ttfamily
  2210.05532}}].

\end{thebibliography}
\end{document}